\documentclass[twocolumn]{autart_arXiv}   
\usepackage{graphicx}          
\usepackage[utf8]{inputenc}

\let\theoremstyle\relax
\usepackage{amsthm}
\usepackage[noadjust]{cite}

\makeatletter\renewcommand{\@listI}{\leftmargin\leftmargini\labelwidth\labelwidthi\labelsep\normal@labelsep\topsep7pt\partopsep0pt\parsep0pt\itemsep0pt\parskip0pt}\makeatother

\usepackage{enumitem}
\setdescription{leftmargin=0pt}

\usepackage{comment}

\usepackage{amsmath,amssymb, color,booktabs}
\allowdisplaybreaks

\usepackage{tikz}
\usetikzlibrary{shapes,arrows,calc,positioning}
\definecolor{MyGreen}{RGB}{50,140,80}
\usetikzlibrary{arrows.meta}
\usepackage{pgfplots}
\definecolor{mycolor1}{rgb}{0.00000,0.44700,0.74100}
\definecolor{mycolor2}{rgb}{0.85000,0.32500,0.09800}

\makeatletter
\def\thm@space@setup{%
  \thm@preskip=\parskip
  \thm@postskip=0pt
}
\makeatother

\begin{document}

\newcommand{\Rl}[2]{\ensuremath{\mathbb{R}^{#1}_{#2}}}   
\newcommand{\R}{\ensuremath{\mathbb{R}}}
\newcommand{\Rlp}{\ensuremath{\mathbb{R}_{>0}}}
\newcommand{\Rlpc}{\ensuremath{\overline{\mathbb{R}}_{>0}}}
\newcommand{\Rlo}{\ensuremath{\mathbb{R}_{\geq 0}}}
\newcommand{\Rln}{\ensuremath{\mathbb{R}_{< 0}}}
\newcommand{\Zo}{\ensuremath{\mathbb{Z}_{\geq 0}}}
\newcommand{\Zp}{\ensuremath{\mathbb{Z}_{> 0}}}
\newcommand{\C}{\ensuremath{\mathbb{C}}}
\newcommand{\N}{\ensuremath{\mathbb{N}}}      
\newcommand{\No}{\ensuremath{\mathbb{N}_{\geq 0}}}
\newcommand{\Np}{\ensuremath{\mathbb{N}_{> 0}}}
\newcommand{\Z}{\ensuremath{\mathbb{Z}}}

\definecolor{bleucit}{rgb}{0.2,0.4,0.6} 
\newcommand{\bleucit}{\textcolor{bleucit}}
\newcommand\cp[1]{{`\emph{#1}'}}
\newcommand\blue[1]{\emph{{\color{blue}#1}}}

\newcommand{\cmark}{\ding{51}}%
\newcommand{\xmark}{\ding{55}}%

\newcommand{\id}{\ensuremath{\text{id}}}

\newcommand{\dst}{\displaystyle}
\newcommand{\Linf}[1]{\ensuremath{\mathcal{L}^{#1}}}

\newcommand{\eg}{{\it e.g.}}

\newcommand{\Nesic}{Ne{\v{s}}i{\'c} }

\definecolor{blue_cv}{rgb}{0.09,0.35,0.78}

\newcommand{\KL}{\ensuremath{\mathcal{KL}}}
\newcommand{\K}{\ensuremath{\mathcal{K}}}
\newcommand{\Kinf}{\ensuremath{\mathcal{K}_{\infty}}}
\newcommand{\KK}{\ensuremath{\mathcal{KK}}}
\newcommand{\KN}{\ensuremath{\mathcal{KN}}}
\newcommand{\KKL}{\ensuremath{\mathcal{KKL}}}
\newcommand{\KLL}{\ensuremath{\mathcal{KLL}}}
\newcommand{\D}{\ensuremath{\mathcal{D}}}
\newcommand{\PD}{\ensuremath{\mathcal{PD}}}

\newcommand{\Cs}{\ensuremath{C_{\text{steady}}}}
\newcommand{\Ct}{\ensuremath{C_{\text{transient}}}}
\newcommand{\Ds}{\ensuremath{D_{\text{steady}}}}
\newcommand{\Dt}{\ensuremath{D_{\text{transient}}}}
\newcommand{\UtGpAS}{U$_{\text{t}}$GpAS}
\newcommand{\UtGAS}{U$_{\text{t}}$GAS}
\newcommand{\UjGpAS}{U$_{\text{j}}$GpAS}
\newcommand{\UjGAS}{U$_{\text{j}}$GAS}

\newcommand{\argmin}{\ensuremath{\text{argmin}\,}}
\newcommand{\interior}{\ensuremath{\text{int}\,}}
\newcommand{\dom}{\ensuremath{\text{dom}\,}}
\newcommand{\Span}{\ensuremath{\text{Span}}}
\newcommand{\avg}{\ensuremath{\text{avg}}}
\newcommand{\co}{\ensuremath{\text{co}\,}}
\newcommand{\coc}{\ensuremath{\overline{\text{dom}}\,}}
\newcommand{\ext}{\ensuremath{\text{ext}}}
\newcommand{\rge}{\ensuremath{\text{rge}\,}}
\newcommand{\esup}{\ensuremath{\text{ess.sup}\,}}

\newcommand{\sign}[1]{\ensuremath{\text{sign}{(#1)}}}
\newcommand{\sat}{\ensuremath{\text{sat}}}

\newcommand{\sinc}{\ensuremath{\text{sinc}}}
\newcommand{\nom}{\ensuremath{\text{nom}}}

\newcommand{\Tmati}{\ensuremath{T_{MATI}\,\,}}
\newcommand{\Tmasp}{\ensuremath{\mathrm{T_{MASP}}\,}}
\newcommand{\lc}{\ensuremath{\llbracket}}
\newcommand{\rc}{\ensuremath{\rrbracket}}

\newcommand{\norm}[1]{\ensuremath{\left\|{#1}\right\|}}
\newcommand{\ip}[2]{\ensuremath{\left\langle #1, #2\right\rangle}}
\newcommand{\cb}[1]{\ensuremath{\overline{\mathbb{B}}_{\mbox{\scriptsize $#1$}}}}                              
\newcommand{\ob}[2]{\ensuremath{\mathbb{B}_{\mbox{\scriptsize $#1$}}\ensuremath{\left( #2\right)}}}   
\newcommand{\df}{\ensuremath{\stackrel{\mbox{\tiny $\mathrm{def}$}}{=}\:}}                                             
\newcommand{\myint}[4]{\ensuremath{\int_{#1}^{#2}#3\;\mathrm{d}#4}}
\newcommand{\Mm}{\ensuremath{\:\stackrel{\rightarrow}{\scriptstyle{\rightarrow}}\:}}
\newcommand{\bm}[1]{\ensuremath{\mathbf{#1}}}
\newcommand{\hs}[1]{\hspace*{#1 em}}
\newcommand{\qa}{\ensuremath{\mathcal{Q}_{A}}}%
\newcommand{\mc}[1]{\ensuremath{\mathcal{#1}}}
\newcommand{\di}{\ensuremath{\mathcal{D}_{i}}}

\newcommand{\HS}{\ensuremath{\mathcal{H}}}
\newcommand{\HSc}{\ensuremath{\mathcal{H}_c}}


%


\newtheorem{exple}{Example}
\newtheorem{ass}{\textnormal{\textbf{Assumption}}}
\newtheorem{propos}{Proposition}
\newtheorem{lemma}{Lemma}
\newtheorem{ex}{Example}
\newtheorem{theorem}{\textnormal{\textbf{Theorem}}}
\newtheorem{corol}{\textnormal{\textbf{Corollary}}}
\newtheorem{pb}{Problem}
\newtheorem{remark}{\textnormal{\textbf{Remark}}}
\newtheorem{sass}{Standing Assumption}
\newtheorem{obj}{\textnormal{\textbf{Objective}}}


%
%
\newenvironment{rems}{\textit{Remarks. }}{\mbox{}\\[1ex]}

\newenvironment{eqn}[1][]{%
    \ifx&#1&\else\label{#1}\fi%
    \equation
    \setlength{\arraycolsep}{1.5pt} 
    \begin{array}{rlllll}%
}{%
    \end{array}\endequation%
}

\newcommand{\f}{\ensuremath{\text{f}}}
\newcommand{\zf}{\ensuremath{\text{Zf}}}
\newcommand{\p}{\ensuremath{\text{p}}}
\newcommand{\cont}{\ensuremath{\text{c}}}
\newcommand{\continuous}{\ensuremath{\text{cont}}}
\newcommand{\n}{\ensuremath{\text{n}}}
\newcommand{\aug}{\ensuremath{\text{aug}}}

\newcommand{\rank}{\ensuremath{\text{rank}}}

\begin{frontmatter}

\title{Robust stabilization with spiking neuronal communication}

\thanks{Work  partly funded  by the CNRS International Research Project NEUROCON and the ANR under grant OLYMPIA ANR-23-CE48-0006.}

\author[Eindhoven]{Elena Petri}\ead{e.petri@tue.nl},    
\author[Nancy]{Romain Postoyan}\ead{romain.postoyan@univ-lorraine.fr},  
\author[Eindhoven]{Erik Steur}\ead{e.steur@tue.nl},               
\author[Eindhoven]{W.P.M.H. (Maurice) Heemels}\ead{m.heemels@tue.nl}

\address[Eindhoven]{Department of Mechanical Engineering, Eindhoven University of Technology, The Netherlands}  
\address[Nancy]{Universit\'e de Lorraine, CNRS, CRAN, 54000 Nancy, France}             

\begin{keyword}                           
Neuromorphic engineering; Hybrid dynamical systems; Robust stabilization; Input-to-state stability              
\end{keyword}                             

\begin{abstract} 
Neuromorphic engineering develops hardware and software systems inspired by biological neurons, with the goal of achieving energy-efficient, low-latency, robust, and adaptive computation, communication and control. Its potential impact on systems and control is significant, as it may enable novel approaches to control and estimation by leveraging brain-inspired computation and communication principles. In this context, we present a framework for the robust stabilization of a plant subject to disturbances when the communication between noisy sensors and the controller relies on spiking signals generated by neuron-inspired schemes. 
The communication scheme consists of a spike encoder on the sensors side, which is based on integrate-and-fire neurons that convert the analog plant output measurement into a spiking signal, and a spike decoder on the controller side inspired by synaptic processing to convert the received spiking signal into an analog signal.
We  provide design conditions on the spike decoder, the spike encoder as well as on the controller under which the closed-loop system exhibits a practical input-to-state stability property, where the adjustable parameters are the amplitudes of the spikes. The results are shown to be applicable to a class of nonlinear systems as well as to any stabilizable and detectable linear time-invariant system. Numerical simulations on a single-link manipulator illustrate the potential of the approach.
\end{abstract}

\end{frontmatter}

\section{Introduction}\label{sect:introduction}

Inspired by the event-based spiking activity of neurons in the brain and nervous system, neuromorphic engineering aims at developing novel asynchronous hardware and software systems that emulate principles of biological neural computation, communication, and control. This research field is motivated by the potential advantages of brain-inspired technologies over classical digital technologies, including energy-efficiency, low latency, robustness, and adaptability \cite{mead1990neuromorphic}. 
In the field of systems and control, neuromorphic control focuses on the design and analysis of closed-loop systems in which information is processed, communicated, and/or measured using event-based spiking signals. The first pioneering works in this direction are \cite{deweerth1990neuron, deweerth1991simple}, where spikes are used to regulate the speed of an electric motor, providing clear practical advantages over classical control schemes. The theoretical advantages of asynchronous impulsive control over periodic control have been illustrated later in, e.g., \cite{aastrom1999comparison, meng2012optimal}. More recently, neuromorphic control has gained a renewed interest, see, e.g., \cite{sepulchre2022spiking, eilers2025, medvedeva2025formalizing, schmetterling2024neuromorphic, petri2024analysis, petri2025emulation, petri2025onTwo, petri2025rhythmic, heemels2026spiking, klip2026stabilizing,jang2025note}. Despite these recent advances, neuromorphic control is still in its infancy, mainly because systematic analysis and design methodologies are lacking. 

An important research direction in neuromorphic control is how to control dynamical systems when communication between sensors and controller is carried out using spiking signals, inspired by the spiking  mechanism of the biological neural system. 
This approach is motivated by the advantages of spiking communication compared to conventional strategies notably in terms of energy-efficiency and  latency. Indeed, communications take the form of fixed-amplitude spikes  so that information is encoded exclusively in the timing of the spikes. Hence the communication are intrinsically event-based and asynchronous. This helps reducing the latency compared to time-triggered strategies    
and eliminates the need for clocks and for transmitting signals of varying amplitudes. On the other hand, neuromorphic hardware typically consumes significantly less power than standard digital processors see, e.g., \cite{merolla2014million}. 

The event-based nature of spiking neuronal communication reminds of event-triggered control literature\footnote{Interestingly, the seminal work that motivated event-triggered control research \cite{aastrom1999comparison} drew its inspiration from spiking neurons and the controller signal consisted of (Dirac) spikes.}, see, e.g., \cite{heemels2012introduction, postoyan2025event,tabuada2007event} and the references therein. 
The main difference, and advantage, of spiking neuronal communication for control is that it does not require the signal value to be communicated at each transmission instant as already mentioned. This may ease the implementation of the control loop and is expected to provide greater robustness properties, as we show  (see Remark \ref{rem:robustness-to-noise} and the example in Section \ref{sect:example}). Spiking neuronal communication share similarities with  event-triggered techniques based on $1$-bit communication, in which information is also solely encoded in the timing of the events (see, e.g., \cite{mazo-cao-aut14,almakhles-abdelrahim-tac23}). However, $1$-bit event-triggered control techniques do not rely on  spiking communication  and they  either require the receiver to have some knowledge of the coded signal or need the sender and receiver to synchronously update some variables, which may be difficult to ensure in practice. 
In contrast, spiking neuronal communication does not require such information or synchronization.

In this context, we focus on the robust stabilization of nonlinear systems using spiking neuronal communication between the plant sensors and the controller. The objective is to give design conditions on the controller and the spiking neuronal communication scheme to ensure that the closed-loop system exhibits desirable stability properties. 
In particular, a spike encoder, inspired by the neuromorphic controller proposed in \cite{petri2025emulation} and modeled using the integrate-and-fire neuronal model, see, e.g., \cite{abbott1999lapicque, izhikevich2010hybrid}, is proposed to convert the analog, noisy measured output from the plant into a sequence of spikes, modeled as a train of Dirac delta pulses, similarly to,  e.g., \cite{petri2024analysis,petri2025emulation,eilers2025, petri2025rhythmic, boerlin2013predictive, agliati2026spiking,moser2024integrate, gerstner2014neuronal}. The generated spikes are then decoded at the controller side using a spike decoder, which is inspired by neuronal synaptic processing see, e.g., \cite[Chapter 3.1]{gerstner2014neuronal}, \cite{roth2009modeling}. Specifically, the spike decoder filters the spike train to produce an analog signal that is then processed by the controller.

We model the closed-loop system as a hybrid dynamical system affected by exogenous disturbance and measurement noise. To proceed with the design conditions and analysis, we augment the closed-loop model with artificial auxiliary variables that allow interpreting the obtained augmented closed-loop model as the feedback interconnection of a purely continuous-time system with a spiking system. We exploit this viewpoint to derive  conditions on the controller, the spike encoder, and the spike decoder under which the original closed-loop system satisfies a practical input-to-state stability property, where the adjustable parameters are the spike amplitudes, which, in our design, will be chosen, without loss of generality, with the same values as the neurons' thresholds. We show that the required conditions can always be met when the plant to be controlled is a stabilizable and detectable linear time-invariant (LTI) system. We also provide easier-to-test properties for nonlinear systems under which the required conditions hold. Finally, we illustrate the relevance of the approach on the set-point stabilization of a $4^{\text{th}}$-order single-link manipulator. The numerical simulations confirm the expected trade-off between the ultimate bound on the closed-loop solutions and the amount of spikes/communications. It also suggests that spiking communications may help improving the robustness properties of continuous-time communications.

Compared to \cite{petri2025emulation}, our approach is not to emulate a given controller, but focuses on the design of both the controller and the neuronal spiking communication scheme to practically stabilize the system. Importantly, the results are developed for nonlinear systems as opposed to LTI systems in \cite{petri2025emulation}. Compared to \cite{eilers2025}, the approach applies to a larger class of systems (not necessarily fully actuated, and not necessarily affine in the input) and the neuronal communication scheme is different for the same reason as for \cite{petri2025emulation}. The use of neuronal spiking communications for synchronization of multi-agent systems and of the phase of nonlinear oscillators was recently studied in \cite{jang2025note} and \cite{song2026distributed}, respectively. In contrast, we consider set-point stabilization, 
use a spike decoder, which is essential to address general nonlinear plant models subject to disturbances.

The rest of the paper is organized as follows. After some preliminaries in Section \ref{sect:preliminaries}, the problem is described in Section \ref{sect:problem-statement}. We present the spike decoder,  the spike encoder and the controller in Section \ref{sect:design}. After augmenting the original closed-loop model in Section \ref{sect:loop-transformations}, we analyze the stability of the closed-loop system in Section \ref{sect:stability}. We elaborate on how to design the controller, the spike encoder, and the spike decoder to ensure the satisfaction of a key assumption in Section \ref{sect:Sigma-continuous-iss} and provide an illustrative example in Section~\ref{sect:example}. Lengthy proofs and auxiliary intermediate results are postponed to the appendix.

\section{Preliminaries}\label{sect:preliminaries}
 
\noindent\emph{Sets.} The symbol $\R$ stands for the set of real numbers, and $\Rlo$ ($\Rlp$) for the set of non-negative (positive) real numbers.  The symbol $\No$ ($\Np$) denotes the set of non-negative (positive) integers. Given $J\in\No\cup\{\infty\}$, we write $\N_{\leq J}:=\{0,1,\ldots,J\}$ so that $\N_{\leq\infty}=\No$, and we adopt the convention $\N_{\leq -1}=\emptyset$. The field of complex numbers is denoted $\C$.  We define the set of (in)finite sequences of  increasing, non-negative real numbers initialized at $0$ as  $\mathcal{T} := \cup_{J\in\No \cup \{\infty\}} \mathcal{T}_J$ with
$\mathcal{T}_J  := \big\{ \{t_k\}_{k\in\N_{\leq J}}\in (\mathbb{R}_{\ge 0})^{\N_{\leq J}}:\, t_0=0,\,\forall k \in \N_{J-1}\backslash\{0\}\,\, \, t_k < t_{k+1}\big\}$ for $J\in\No\cup\{\infty\}$; we denote
$\mathcal{T}_{\zf}^{\infty}$ the set of such sequences that do not exhibit accumulation points, so ``Zeno''-free, (including finite sequences), i.e.,  $\mathcal{T}_ {\zf}^\infty := \big\{ \{t_k\}_{k\in\N_{\leq J}}\in \mathcal{T} \text{ with } J\in\N\cup\{\infty\}:\,t_k\to\infty \text{ as } k\to\infty \text{ if } J=\infty\big\}$.  
%
%
%
\smallskip \newline\noindent\emph{Matrices and vectors.} The notation $I_n$ denotes the identity matrix of $\R^n$, and $0_{n\times m}$ the zero matrix of $\R^{n\times m}$ or simply $0$ when the dimension is clear from the context, $n,m\in\Np$. Given a vector $x\in\R^n$ with $n\in\Zp$, $\text{diag}(x)$ stands for the $n\times n$ diagonal matrix whose diagonal elements correspond to $x$. The rank of a matrix is denoted $\rank(\cdot)$. Given any two vectors $x\in\R^{n_\p}$ and $y\in\R^{n_y}$ with $n_\p,n_y\in\Np$, $(x,y)$ stands for $(x^\top,y^\top)^\top$. The real part of $z\in\C$ is denoted $\Re(z)$.
%
%
\smallskip \newline\noindent\emph{Norms.} For a vector $x \in \R^n$, $|x|$ denotes its Euclidean norm and $|x|_{\mathcal{A}}$  its distance to a non-empty set $\mathcal{A}\subset\R^{n}$, i.e., $|x|_{\mathcal{A}}=\inf\{|x-z|\,:\,z\in\mathcal{A}\}$. For a matrix $A \in \R^{n  \times m}, \norm{A}$ stands for its induced 2-norm. For a Lebesgue measurable signal $v: \R_{\geq0} \to \R^{n_v}$ with $n_v \in \Np$, and $t_1,t_2\in\R_{\geq 0}\cup\{\infty\}$ with $t_1\leq t_2$,  $\norm{v}_{[t_1, t_2]}:= \textnormal{ess.} \sup_{t \in [t_1, t_2]} |v(t)|$. 
%
%
\smallskip \newline \noindent\emph{Functions.} The identity map from a set to itself is denoted $\id$. We consider sets  $\Kinf$ and $\KL$ for comparison functions as defined in \cite[Chap. 3]{goebel2012hybrid}. We write that $\beta\in\exp$-$\KL$ when $\beta:\Rlo\times\Rlo\to\Rlo$ and there exists $(c,\lambda)\in(0,1)\in[1,\infty)$ such that  $\beta:(s,t)\mapsto\lambda c^t s$. 
Given $c\in\R$, we introduce the maps $\mathbf{1}_{\geq c}:\R\to\R$, $\mathbf{1}_{\geq c}(z)=1$ if $z\geq c$ and $\mathbf{1}_{\geq c}(z)=0$ if $z< c$, and $\mathbf{1}_{\leq c}:\R\to\R$ with $\mathbf{1}_{\leq c}(z)=1-\mathbf{1}_{\geq c}(z)$ for any $z\in\R$.  Given a function $f:\R_{\geq0} \to \R$, for any $t \geq 0$ we denote by $f( t^+)$ the right limit of $f$ at $t$, i.e., $f(t^+) = \lim_{s\downarrow t}f(s)$ whenever it exists. 
Given a set $\mathcal{Y} \subseteq \R^{n_y}$  with $n_y \in \Np$, $\mathcal{L}_{\mathcal{Y}}$ denotes the set of all functions from $\R_{\geq0}$ to $\mathcal{Y}$ that are Lebesgue measurable and locally essentially bounded.  
%
%
\smallskip \newline\noindent\emph{Distributions and spikes.} Given $\tau\in\Rlo$, $\delta_{\tau}$ is the Dirac measure centered at $\tau \in \Rlo$, which is defined by its action on any continuous function $\phi :\Rlo\to  \R^n$ with any given $n\in\Np$, as $\int_{\Rlo} \phi(t) \delta_{\tau}\text{dt} = \phi(\tau)$.  We refer to $\delta_\tau$ for some $\tau\in\Rlo$ as a \emph{Dirac pulse} or a \emph{(unitary) spike} for the sake of convenience. 
Given $m \in \mathbb{Z}_{>0}$,  $\mathcal{R}_m$ denotes the space of $\mathbb{R}^m$-valued Radon measures on $\Rlo$. We define the set of \emph{trains of spikes} as $
\mathcal{D}_m := \big\{ u \in \mathcal{R}_m\,:\, \exists J\in\No\cup\{\infty\},\,\exists \{A_k\}_{k\in\N_{\leq J}} \in (\mathbb{R}^m)^{\No}, \, \exists \{t_k\}_{k\in\N_{\leq J}} \in \mathcal{T}, \, u = \sum_{k=0}^{J} A_k \delta_{t_k} \big\}$, and the set of \emph{Zeno-free trains of spikes}\footnote{Zeno-free refers here to the fact that an infinite amount of spikes in a finite amount of time cannot occur.} $\mathcal{D}_{m,\infty}$ by replacing $\mathcal{T}$ by $\mathcal{T}_{\zf}^\infty$ in the definition of $\mathcal{D}_m$. Given any element in $\mathcal{D}_m$, the  associated time instants in $\mathcal{T}$ are called \emph{spiking times} (with some slight abuse as $0$ may not be a spiking time as we will see) and we call the associated sequence $\{A_k\}_{k\in\N_{\leq J}}$ the \emph{spike amplitudes}. Let $\mathcal{S}_{m}:=\{v = v_1 +v_2 \,:\,v_1 \in \mathcal{L}_{\R^{m}},\,v_{2}\in\mathcal{D}_{m,\infty},\,\norm{v}_\star<\infty\}$ with $\norm{v}_\star = \sup_{t\geq0}|\int_{0}^{t} v(s) ds|$, which defines a normed space\footnote{See \cite[Definition 1 and Lemma 1]{petri2025emulation} for more details.}. We call any element $v \in \mathcal{S}_m$ a \emph{spiking signal}.
%
%
%
%
\smallskip \newline \noindent\emph{Input-to-state stability \cite{sontag2008input, mironchenko2023-iss-book}.} The continuous-time system $\dot x = f(x,d_1,\ldots,d_q)$ with state $x\in\R^n$ and inputs $d_1\in\R^{m_1},\,\ldots,\,d_q\in\R^{m_q}$ at time $t \in \R_{\geq0}$, $n,m_i\in\Np$ for $i\in\{1,\ldots,q\}$, is \emph{input-to-state stable (ISS) with respect to inputs $d_1,\ldots,d_q$ (with gains $\gamma_1,\ldots,\gamma_q$),} if there exist $\beta\in\KL$ and\footnote{$0$ means her the map from $\Rlo$ to $\Rlo$ with value $0$.} $\gamma_1,\ldots,\,\gamma_q\in\Kinf\cup\{0\}$ such that for any input $d_i\in\mathcal{L}_{\R^{m_i}}$, $i\in\{1,2,\ldots,q\}$, any solution\footnote{Only maximal solutions are considered for the studied dynamical systems in this work, i.e., solutions whose domain of definition cannot be extended.}
with inputs $d_1,\ldots,\,d_q$ satisfies $|x(t)|\leq \max\{\beta(|x(0)|,t),\gamma_1(\|d_1\|_{[0,t]}),\ldots,\gamma_q(\|d_q\|_{[0,t]})\}$ for all $t$ in the domain of the solution. 

\section{Problem description} \label{sect:problem-statement} 

We consider a nonlinear plant modeled as 
\begin{eqn}
   \label{eq:plant}
     \dot x_\p  =  f_\p(x_\p,u,d), & &
     y  =  h_\p(x_\p,w), 
\end{eqn}
where $x_\p\in \R^{n_\p}$ is the state, $u \in \R^{n_u}$ is the control input, $d\in \R^{n_d}$ is the disturbance acting on the dynamics,  $y\in\R^{n_y}$ is the measured output affected by measurement noise $w\in\R^{n_w}$ with $n_\p, n_u, n_y\in\Np$ and $ n_d,n_w \in \No$. The functions $f_\p:\R^{n_\p}\times\R^{n_u}\times\R^{n_d}\to\R^{n_\p}$ and $h_p:\R^{n_\p} \times \R^{n_w}\to\R^{n_y}$ are assumed to be continuous. 

The goal is to stabilize the origin of system (\ref{eq:plant}) while communicating the output $y$ to the controller (to be designed) using  spiking communications generated by a neuron-inspired architecture, as illustrated in Fig. \ref{fig:closed-loop}. 
Specifically, a spike encoder converts the analog output signal $y$ into a spiking signal denoted $y_s$  using a mechanism inspired by neuronal membrane potential dynamics. The spiking signal  $y_s$ is then converted to an analog signal denoted $\widehat y$ using a spike decoder inspired by synaptic processing.  As a consequence, the controller has access to $\widehat y$ and uses it  to stabilize the origin of the closed-loop system as it does not have access to $y$. 

The objectives are  to give design conditions on the spike encoder, the spike decoder and the controller so that the closed-loop system depicted in Fig. \ref{fig:closed-loop} exhibits stability properties as formalized in the sequel and does not exhibit Zeno behavior, i.e., an infinite number of spikes in finite time. 

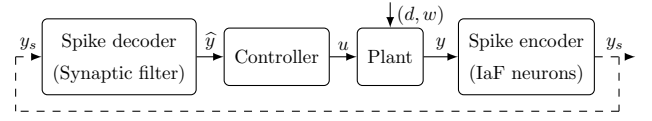
\begin{figure}[!t]
	\begin{center}
		\tikzstyle{blockB} = [draw, fill=blue!30, rectangle, rounded corners=2pt, 
		minimum height=2em, minimum width=3em]  
		\tikzstyle{blockG} = [draw,  rounded corners=2pt,fill=MyGreen!40, rectangle, 
		minimum height=2em, minimum width=3em, rounded corners=2pt]
		\tikzstyle{blockR} = [draw, fill=red!40, rectangle, 
		minimum height=2em, minimum width=3em, rounded corners=2pt]
		\tikzstyle{blockO} = [draw,minimum height=1.5em, fill=orange!20, minimum width=4em, rounded corners=2pt]
		\tikzstyle{input} = [coordinate]
		\tikzstyle{blockW} = [draw,minimum height=3em, fill=white!20, minimum width=2.5em, rounded corners=2pt]
		\tikzstyle{input1} = [coordinate]
		\tikzstyle{blockCircle} = [draw, circle]
		\tikzstyle{sum} = [draw, circle, minimum size=.3cm]
		\tikzstyle{blockSensor} = [draw, fill=white!20, draw= blue!80, line width= 0.8mm, minimum height=10em, minimum width=13em, rounded corners=2pt]
		
		\begin{tikzpicture}[auto, node distance=2cm,>=latex , scale=0.72,transform shape] 
			
			\node [input, name=stateFeedbackSpiky] {};

          \node [blockW, right of=stateFeedbackSpiky, node distance=1.9cm] (SAconverter) { 
				$\begin{array}{c}
					\text{Spike decoder}\\
                    \text{(Synaptic filter)}
				\end{array}$
			};
			
			\draw [draw,dashed, ->] (stateFeedbackSpiky) -- node [pos=0.5]{$y_s$} (SAconverter);
			
			\node [blockW, right=0.5cm of SAconverter] (Controller) { 
				$\begin{array}{c}
					\text{Controller}
				\end{array}$
			};
			
			\draw [draw,->] (SAconverter) -- node [pos=0.5]{$\widehat y$} (Controller);

			\node [blockW, right=0.5cm of Controller] (system) { 
				$\begin{array}{c}
					\text{Plant}
				\end{array}$
			};
			\node [input, above of= system, node distance=1cm] (PlantDisturbance){};
\draw [draw, ->] (PlantDisturbance) -- node [pos=0.5, right]{$(d,w)$} (system);
			\draw [draw,->] (Controller) -- node [pos=0.5]{$u$} (system);

 \node [blockW, right of=system, node distance=2.5cm] (ASconverter) { 
				$\begin{array}{c}
					\text{Spike encoder}\\
                    \text{(IaF neurons)}
				\end{array}$
			};
			
			\draw [draw,->] (system) -- node [pos=0.5]{$y$} (ASconverter);

			\node [input, right of= ASconverter, node distance=2.0cm] (SystemOutputSpiky){};
			\node [input, right of= ASconverter, node distance=1.7cm] (SystemFeedbackSpiky){};
			\node [input, below of= SystemFeedbackSpiky, node distance=1.0cm] (SystemFeedbackSpikyBelow){};
            \node [input, left of= SystemFeedbackSpikyBelow, node distance=3.8cm] (NetworkIn){};
            \node [input, left of= NetworkIn, node distance=3.2cm] (NetworkOut){};
			\node [input, below of= stateFeedbackSpiky, node distance=1.0cm] (stateFeedbackSpikyBelow){};
			
			\draw [draw,dashed, ->] (ASconverter) -- node [pos=0.5]{$y_s$} (SystemOutputSpiky);
			\draw [draw,dashed] (SystemFeedbackSpiky) --  (SystemFeedbackSpikyBelow);
            \draw [draw,dashed] (stateFeedbackSpiky) --  (stateFeedbackSpikyBelow);
	          \draw [draw,dashed] (SystemFeedbackSpikyBelow) --  (stateFeedbackSpikyBelow);

            \node [input, right of= Controller, node distance=1.2cm] (ControllerRight){};

		\end{tikzpicture}
	\end{center}
	\caption{Closed loop with spiking neuronal communications (IaF: integrate-and-fire).}
	\label{fig:closed-loop}
\end{figure}

\section{Spike encoder, spike decoder and controller}\label{sect:design}

We present the spike encoder  (Section~\ref{subsect:spiking-encoder}), the spike decoder (Section~\ref{subsect:spiking-decoder}), as well as the controller (Section~\ref{Controller}), before deriving the overall closed-loop model (Section~\ref{subsect:closed-loop-model}).

\subsection{Spike encoder}\label{subsect:spiking-encoder}

Writing the output signal $y=(y_1, y_2, \dots, y_{n_y})$  with $y_i(t)\in\R$ for any $t \in \R_{\geq0}$ and any $i\in\{1,\ldots,n_y\}$, the spike encoder  transforms each scalar analog component $y_i$ of $y$ to a spiking signal   $y_{s,i}$ thereby forming  $y_s=(y_{s,1},\ldots,y_{s,n_y})$ with $y_{s,i}(t)\in\R$ for any $t \in \R_{\geq0}$ and any $i\in\{1,\ldots,n_y\}$.  
Like in \cite{petri2025emulation},  $y_{s,i}$, $i \in \{1,\ldots, n_y\}$, is generated by $2$ neurons, whose dynamics are inspired by the integrate-and-fire model \cite{abbott1999lapicque, izhikevich2010hybrid}. Therefore, the spike encoder uses a total of $2n_y$ neurons to transform the analog signal $y$ into the spiking signal  $y_s$; see Fig.~\ref{Fig:blockDiagram_ASconvMultipleOutput} for an illustration. We present the dynamics of the spike encoder below.

\begin{figure}[!t]
	\begin{center}
		\tikzstyle{blockB} = [draw, fill=blue!30, rectangle, 
		minimum height=2em, minimum width=3em, rounded corners=2pt]  
		\tikzstyle{blockG} = [draw, fill=MyGreen!40, rectangle, 
		minimum height=2em, minimum width=3em, rounded corners=2pt]
		\tikzstyle{blockR} = [draw, fill=red!40, rectangle, 
		minimum height=2em, minimum width=3em, rounded corners=2pt]
		\tikzstyle{blockO} = [draw,minimum height=1.5em, fill=orange!20, minimum width=4em, rounded corners=2pt]
		\tikzstyle{input} = [coordinate]
		\tikzstyle{blockW} = [draw,minimum height=1.0em, fill=white!20, minimum width=2.0em, rounded corners=2pt]
        \tikzstyle{blockWneurons} = [draw,minimum height=1.0em, fill=white!20, minimum width=10em, rounded corners=2pt]
		\tikzstyle{input1} = [coordinate]
		\tikzstyle{blockCircle} = [draw, circle]
		\tikzstyle{sum} = [draw, circle, minimum size=.3cm]
		\tikzstyle{blockSensor} = [draw, fill=white!20, draw= blue!80, line width= 0.8mm, minimum height=10em, minimum width=13em, rounded corners=2pt]
		
		\begin{tikzpicture}[auto, node distance=2cm,>=latex , scale=0.68,transform shape] 
			
			\node [input, name=stateFeedback] {};
			\node [input, right of= stateFeedback, node distance=1cm] (stateFeedbackRight){};
			\node [input, above of= stateFeedbackRight, node distance=0.65cm] (stateFeedbackUp){};
			\node [input, below of= stateFeedbackRight, node distance=0.65cm] (stateFeedbackDown){};
			
			\node [blockWneurons, right of=stateFeedbackUp, node distance=2.3cm] (Neuron1) { 
				%
				$\begin{array}{c}
                  \vspace{-0.7em}
					\text{Neuron } (1,1)\\  (\xi_{1,1}, \Delta_{1,1}, \alpha_{1,1})
				\end{array}$
			};
			\node [blockWneurons, right of=stateFeedbackDown, node distance=2.3cm] (Neuron2) { 
				$\begin{array}{c}
                \vspace{-0.7em}
					\text{Neuron } (1,2)\\  (\xi_{1,2}, \Delta_{1,2}, \alpha_{1,2})
				\end{array}$
			};
			
			\draw [draw,-] (stateFeedback) -- node [pos=0.5]{$y_1$} (stateFeedbackRight);
			\draw [draw,-] (stateFeedbackRight) --  (stateFeedbackUp);
			\draw [draw,-] (stateFeedbackRight) --  (stateFeedbackDown);
			\draw [draw,->] (stateFeedbackUp) -- node {} (Neuron1);
			\draw [draw,->] (stateFeedbackDown) -- node {} (Neuron2);
			
			\node [input, right of= Neuron1, node distance=4.5cm] (Neuron1End){};
			\node [input, right of= Neuron2, node distance=4.5cm] (Neuron2End){};
			\node [blockW, right of=Neuron2, node distance=3.8cm] (Neuron2Gain) { 
				$-1$
			};
			
			\draw [draw,-] (Neuron1) -- node {} (Neuron1End);
			\draw [draw,->] (Neuron2) -- node {} (Neuron2Gain);
			\draw [draw,-] (Neuron2Gain) -- node {} (Neuron2End);
			
			\node [input, below of= Neuron1End, node distance=0.55cm] (Neuron1Point){};
			\node [input, above of= Neuron2End, node distance=0.55cm] (Neuron2Point){};
			\node [input, below of= Neuron1Point, node distance=0.1cm] (NeuronCircle){};
			\draw [draw,->] (Neuron1End) -- node {} (Neuron1Point);
			\draw [draw,->] (Neuron2End) -- node {} (Neuron2Point);
			
			\draw [fill=white] (NeuronCircle) circle (0.1cm);
			
			\node [input, right of= NeuronCircle, node distance=0.1cm] (NeuronCircleRight){};
			
			\node [input, right of= NeuronCircleRight, node distance=2cm] (SpikySignal){};
			
			 \draw [draw,-] (NeuronCircleRight) -- node [pos=0.8]{$y_{s,1}$} (SpikySignal);
			
			

			\node [input, right of= Neuron1, node distance=2.2cm] (Neuron1SpikeBeginDown){};
			\node [input, above of= Neuron1SpikeBeginDown, node distance=0.2cm] (Neuron1SpikeBegin){};
			\node [input, above of= Neuron1SpikeBegin, node distance=0.45cm] (Neuron1Spike1Up){};
			\node [input, right of= Neuron1SpikeBegin, node distance=0.1cm] (Neuron1Spike2Down){};
			\node [input, left of= Neuron1SpikeBegin, node distance=0.2cm] (Neuron1SpikeBeginLeft){};
			\node [input, above of= Neuron1Spike2Down, node distance=0.45cm] (Neuron1Spike2Up){};
			\node [input, right of= Neuron1Spike2Down, node distance=0.4cm] (Neuron1Spike3Down){};
			\node [input, above of= Neuron1Spike3Down, node distance=0.45cm] (Neuron1Spike3Up){};
			\node [input, right of= Neuron1Spike3Down, node distance=0.3cm] (Neuron1Spike4Down){};
			\node [input, above of= Neuron1Spike4Down, node distance=0.45cm] (Neuron1Spike4Up){};
			\node [input, right of= Neuron1Spike3Down, node distance=0.2cm] (Neuron1Spike4Down){};
			\node [input, above of= Neuron1Spike4Down, node distance=0.45cm] (Neuron1Spike4Up){};
			\node [input, right of= Neuron1Spike4Down, node distance=0.2cm] (Neuron1SpikeFinal){};
			
			\draw [draw,-] (Neuron1SpikeBegin) --  (Neuron1SpikeBeginLeft);
			\draw [draw,-] (Neuron1SpikeBegin) --  (Neuron1Spike1Up);
			\draw [draw,-] (Neuron1SpikeBegin) --  (Neuron1Spike2Down);
			\draw [draw,-] (Neuron1Spike2Down) --  (Neuron1Spike2Up);
			\draw [draw,-] (Neuron1Spike2Down) --  (Neuron1Spike3Down);
			\draw [draw,-] (Neuron1Spike3Down) --  (Neuron1Spike3Up);
			\draw [draw,-] (Neuron1Spike3Down) --  (Neuron1Spike4Down);
			\draw [draw,-] (Neuron1Spike4Down) --  (Neuron1Spike4Up);
			\draw [draw,-] (Neuron1Spike4Down) --  (Neuron1SpikeFinal);
			
			\node [input, right of= Neuron2, node distance=2.2cm] (Neuron2SpikeBeginDown){};
			\node [input, above of= Neuron2SpikeBeginDown, node distance=0.2cm] (Neuron2SpikeBegin){};
			\node [input, above of= Neuron2SpikeBegin, node distance=0.45cm] (Neuron2Spike1Up){};
			\node [input, right of= Neuron2SpikeBegin, node distance=0.3cm] (Neuron2Spike2Down){};
			\node [input, left of= Neuron2SpikeBegin, node distance=0.1cm] (Neuron2SpikeBeginLeft){};
			\node [input, above of= Neuron2Spike2Down, node distance=0.45cm] (Neuron2Spike2Up){};
			\node [input, right of= Neuron2Spike2Down, node distance=0.2cm] (Neuron2Spike3Down){};
			\node [input, above of= Neuron2Spike3Down, node distance=0.45cm] (Neuron2Spike3Up){};
			\node [input, right of= Neuron2Spike3Down, node distance=0.2cm] (Neuron2Spike4Down){};
			\node [input, above of= Neuron2Spike4Down, node distance=0.45cm] (Neuron2Spike4Up){};
			\node [input, right of= Neuron2Spike3Down, node distance=0.3cm] (Neuron2Spike4Down){};
			\node [input, above of= Neuron2Spike4Down, node distance=0.45cm] (Neuron2Spike4Up){};
			\node [input, right of= Neuron2Spike4Down, node distance=0.2cm] (Neuron2SpikeFinal){};
			
			\draw [draw,-] (Neuron2SpikeBegin) --  (Neuron2SpikeBeginLeft);
			\draw [draw,-] (Neuron2SpikeBegin) --  (Neuron2Spike1Up);
			\draw [draw,-] (Neuron2SpikeBegin) --  (Neuron2Spike2Down);
			\draw [draw,-] (Neuron2Spike2Down) --  (Neuron2Spike2Up);
			\draw [draw,-] (Neuron2Spike2Down) --  (Neuron2Spike3Down);
			\draw [draw,-] (Neuron2Spike3Down) --  (Neuron2Spike3Up);
			\draw [draw,-] (Neuron2Spike3Down) --  (Neuron2Spike4Down);
			\draw [draw,-] (Neuron2Spike4Down) --  (Neuron2Spike4Up);
			\draw [draw,-] (Neuron2Spike4Down) --  (Neuron2SpikeFinal);

			\node [input, right of= NeuronCircle, node distance=0.35cm] (Neuron1SpikeBeginDown_tot){};
			\node [input, above of= Neuron1SpikeBeginDown_tot, node distance=0.8cm] (Neuron1SpikeBegin_tot){};
			\node [input, above of= Neuron1SpikeBegin_tot, node distance=0.45cm] (Neuron1Spike1Up_tot){};
			\node [input, right of= Neuron1SpikeBegin_tot, node distance=0.1cm] (Neuron1Spike2Down_tot){};
			\node [input, left of= Neuron1SpikeBegin_tot, node distance=0.2cm] (Neuron1SpikeBeginLeft_tot){};
			\node [input, above of= Neuron1Spike2Down_tot, node distance=0.45cm] (Neuron1Spike2Up_tot){};
			\node [input, right of= Neuron1Spike2Down_tot, node distance=0.4cm] (Neuron1Spike3Down_tot){};
			\node [input, above of= Neuron1Spike3Down_tot, node distance=0.45cm] (Neuron1Spike3Up_tot){};
			\node [input, right of= Neuron1Spike3Down_tot, node distance=0.3cm] (Neuron1Spike4Down_tot){};
			\node [input, above of= Neuron1Spike4Down_tot, node distance=0.45cm] (Neuron1Spike4Up_tot){};
			\node [input, right of= Neuron1Spike3Down_tot, node distance=0.2cm] (Neuron1Spike4Down_tot){};
			\node [input, above of= Neuron1Spike4Down_tot, node distance=0.45cm] (Neuron1Spike4Up_tot){};
			\node [input, right of= Neuron1Spike4Down_tot, node distance=0.2cm] (Neuron1SpikeFinal_tot){};
			
			\draw [draw,-] (Neuron1SpikeBegin_tot) --  (Neuron1SpikeBeginLeft_tot);
			\draw [draw,-] (Neuron1SpikeBegin_tot) --  (Neuron1Spike1Up_tot);
			\draw [draw,-] (Neuron1SpikeBegin_tot) --  (Neuron1Spike2Down_tot);
			\draw [draw,-] (Neuron1Spike2Down_tot) --  (Neuron1Spike2Up_tot);
			\draw [draw,-] (Neuron1Spike2Down_tot) --  (Neuron1Spike3Down_tot);
			\draw [draw,-] (Neuron1Spike3Down_tot) --  (Neuron1Spike3Up_tot);
			\draw [draw,-] (Neuron1Spike3Down_tot) --  (Neuron1Spike4Down_tot);
			\draw [draw,-] (Neuron1Spike4Down_tot) --  (Neuron1Spike4Up_tot);
			\draw [draw,-] (Neuron1Spike4Down_tot) --  (Neuron1SpikeFinal_tot);
			
			\node [input, right of= NeuronCircle, node distance=0.3cm] (Neuron2SpikeBeginDown_tot){};
			\node [input, above of= Neuron2SpikeBeginDown_tot, node distance=0.8cm] (Neuron2SpikeBegin_tot){};
			\node [input, below of= Neuron2SpikeBegin_tot, node distance=0.45cm] (Neuron2Spike1Up_tot){};
			\node [input, right of= Neuron2SpikeBegin_tot, node distance=0.3cm] (Neuron2Spike2Down_tot){};
			\node [input, left of= Neuron2SpikeBegin_tot, node distance=0.1cm] (Neuron2SpikeBeginLeft_tot){};
			\node [input, below of= Neuron2Spike2Down_tot, node distance=0.45cm] (Neuron2Spike2Up_tot){};
			\node [input, right of= Neuron2Spike2Down_tot, node distance=0.2cm] (Neuron2Spike3Down_tot){};
			\node [input, below of= Neuron2Spike3Down_tot, node distance=0.45cm] (Neuron2Spike3Up_tot){};
			\node [input, right of= Neuron2Spike3Down_tot, node distance=0.2cm] (Neuron2Spike4Down_tot){};
			\node [input, below of= Neuron2Spike4Down_tot, node distance=0.45cm] (Neuron2Spike4Up_tot){};
			\node [input, right of= Neuron2Spike3Down_tot, node distance=0.3cm] (Neuron2Spike4Down_tot){};
			\node [input, below of= Neuron2Spike4Down_tot, node distance=0.45cm] (Neuron2Spike4Up_tot){};
			\node [input, right of= Neuron2Spike4Down_tot, node distance=0.2cm] (Neuron2SpikeFinal_tot){};
			
			\draw [draw,-] (Neuron2SpikeBegin_tot) --  (Neuron2Spike1Up_tot);
			\draw [draw,-] (Neuron2Spike2Down_tot) --  (Neuron2Spike2Up_tot);
			\draw [draw,-] (Neuron2Spike3Down_tot) --  (Neuron2Spike3Up_tot);
			\draw [draw,-] (Neuron2Spike4Down_tot) --  (Neuron2Spike4Up_tot);

            \node [input, below of= stateFeedback, node distance=3cm] (stateFeedbackLast){};
            \node [input, right of= stateFeedbackLast, node distance=1cm] (stateFeedbackRightLast){};
			\node [input, above of= stateFeedbackRightLast, node distance=0.65cm] (stateFeedbackUpLast){};
			\node [input, below of= stateFeedbackRightLast, node distance=0.65cm] (stateFeedbackDownLast){};
			
			\node [blockWneurons, right of=stateFeedbackUpLast, node distance=2.3cm] (Neuron1Last) { 
				%
				$\begin{array}{c}
                \vspace{-0.7em}
					\text{Neuron } (n_{y},1)\\  (\xi_{n_{y,1}}, \Delta_{n_{y,1}}, \alpha_{n_{y,1}})
				\end{array}$
			};
			\node [blockWneurons, right of=stateFeedbackDownLast, node distance=2.3cm] (Neuron2Last) { 
				$\begin{array}{c}
                \vspace{-0.7em}
					\text{Neuron } (n_{y},2)\\  (\xi_{n_{y,2}}, \Delta_{n_{y,2}}, \alpha_{n_{y,2}})
				\end{array}$
			};
			
			\draw [draw,-] (stateFeedbackLast) -- node [pos=0.5]{$y_{n_y}$} (stateFeedbackRightLast);
			\draw [draw,-] (stateFeedbackRightLast) --  (stateFeedbackUpLast);
			\draw [draw,-] (stateFeedbackRightLast) --  (stateFeedbackDownLast);
			\draw [draw,->] (stateFeedbackUpLast) -- node {} (Neuron1Last);
			\draw [draw,->] (stateFeedbackDownLast) -- node {} (Neuron2Last);
			
			\node [input, right of= Neuron1Last, node distance=4.5cm] (Neuron1EndLast){};
			\node [input, right of= Neuron2Last, node distance=4.5cm] (Neuron2EndLast){};
			\node [blockW, right of=Neuron2Last, node distance=3.8cm] (Neuron2GainLast) { 
				$-1$
			};
			
			\draw [draw,-] (Neuron1Last) -- node {} (Neuron1EndLast);
			\draw [draw,->] (Neuron2Last) -- node {} (Neuron2GainLast);
			\draw [draw,-] (Neuron2GainLast) -- node {} (Neuron2EndLast);
			
			\node [input, below of= Neuron1EndLast, node distance=0.55cm] (Neuron1PointLast){};
			\node [input, above of= Neuron2EndLast, node distance=0.55cm] (Neuron2PointLast){};
			\node [input, below of= Neuron1PointLast, node distance=0.1cm] (NeuronCircleLast){};
			\draw [draw,->] (Neuron1EndLast) -- node {} (Neuron1PointLast);
			\draw [draw,->] (Neuron2EndLast) -- node {} (Neuron2PointLast);
			
			\draw [fill=white] (NeuronCircleLast) circle (0.1cm);
			
			\node [input, right of= NeuronCircleLast, node distance=0.1cm] (NeuronCircleRightLast){};
			
			\node [input, right of= NeuronCircleRightLast, node distance=2cm] (SpikySignalLast){};
			
			 \draw [draw,-] (NeuronCircleRightLast) -- node [pos=0.8]{$y_{s, n_{y}}$} (SpikySignalLast);
			
			

			\node [input, right of= Neuron1Last, node distance=2.2cm] (Neuron1SpikeBeginDownLast){};
			\node [input, above of= Neuron1SpikeBeginDownLast, node distance=0.2cm] (Neuron1SpikeBeginLast){};
			\node [input, above of= Neuron1SpikeBeginLast, node distance=0.45cm] (Neuron1Spike1UpLast){};
			\node [input, right of= Neuron1SpikeBeginLast, node distance=0.1cm] (Neuron1Spike2DownLast){};
			\node [input, left of= Neuron1SpikeBeginLast, node distance=0.2cm] (Neuron1SpikeBeginLeftLast){};
			\node [input, above of= Neuron1Spike2DownLast, node distance=0.45cm] (Neuron1Spike2UpLast){};
			\node [input, right of= Neuron1Spike2DownLast, node distance=0.4cm] (Neuron1Spike3DownLast){};
			\node [input, above of= Neuron1Spike3DownLast, node distance=0.45cm] (Neuron1Spike3UpLast){};
			\node [input, right of= Neuron1Spike3DownLast, node distance=0.3cm] (Neuron1Spike4DownLast){};
			\node [input, above of= Neuron1Spike4DownLast, node distance=0.45cm] (Neuron1Spike4UpLast){};
			\node [input, right of= Neuron1Spike3DownLast, node distance=0.2cm] (Neuron1Spike4DownLast){};
			\node [input, above of= Neuron1Spike4DownLast, node distance=0.45cm] (Neuron1Spike4UpLast){};
			\node [input, right of= Neuron1Spike4DownLast, node distance=0.2cm] (Neuron1SpikeFinalLast){};
			
			\draw [draw,-] (Neuron1SpikeBeginLast) --  (Neuron1SpikeBeginLeftLast);
			\draw [draw,-] (Neuron1SpikeBeginLast) --  (Neuron1Spike1UpLast);
			\draw [draw,-] (Neuron1SpikeBeginLast) --  (Neuron1Spike2DownLast);
			\draw [draw,-] (Neuron1Spike2DownLast) --  (Neuron1Spike2UpLast);
			\draw [draw,-] (Neuron1Spike2DownLast) --  (Neuron1Spike3DownLast);
			\draw [draw,-] (Neuron1Spike3DownLast) --  (Neuron1Spike3UpLast);
			\draw [draw,-] (Neuron1Spike3DownLast) --  (Neuron1Spike4DownLast);
			\draw [draw,-] (Neuron1Spike4DownLast) --  (Neuron1Spike4UpLast);
			\draw [draw,-] (Neuron1Spike4DownLast) --  (Neuron1SpikeFinalLast);
			
			\node [input, right of= Neuron2Last, node distance=2.2cm] (Neuron2SpikeBeginDownLast){};
			\node [input, above of= Neuron2SpikeBeginDownLast, node distance=0.2cm] (Neuron2SpikeBeginLast){};
			\node [input, above of= Neuron2SpikeBeginLast, node distance=0.45cm] (Neuron2Spike1UpLast){};
			\node [input, right of= Neuron2SpikeBeginLast, node distance=0.3cm] (Neuron2Spike2DownLast){};
			\node [input, left of= Neuron2SpikeBeginLast, node distance=0.1cm] (Neuron2SpikeBeginLeftLast){};
			\node [input, above of= Neuron2Spike2DownLast, node distance=0.45cm] (Neuron2Spike2UpLast){};
			\node [input, right of= Neuron2Spike2DownLast, node distance=0.2cm] (Neuron2Spike3DownLast){};
			\node [input, above of= Neuron2Spike3DownLast, node distance=0.45cm] (Neuron2Spike3UpLast){};
			\node [input, right of= Neuron2Spike3DownLast, node distance=0.2cm] (Neuron2Spike4DownLast){};
			\node [input, above of= Neuron2Spike4DownLast, node distance=0.45cm] (Neuron2Spike4UpLast){};
			\node [input, right of= Neuron2Spike3DownLast, node distance=0.3cm] (Neuron2Spike4DownLast){};
			\node [input, above of= Neuron2Spike4DownLast, node distance=0.45cm] (Neuron2Spike4UpLast){};
			\node [input, right of= Neuron2Spike4DownLast, node distance=0.2cm] (Neuron2SpikeFinalLast){};
			
			\draw [draw,-] (Neuron2SpikeBeginLast) --  (Neuron2SpikeBeginLeftLast);
			\draw [draw,-] (Neuron2SpikeBeginLast) --  (Neuron2Spike1UpLast);
			\draw [draw,-] (Neuron2SpikeBeginLast) --  (Neuron2Spike2DownLast);
			\draw [draw,-] (Neuron2Spike2DownLast) --  (Neuron2Spike2UpLast);
			\draw [draw,-] (Neuron2Spike2DownLast) --  (Neuron2Spike3DownLast);
			\draw [draw,-] (Neuron2Spike3DownLast) --  (Neuron2Spike3UpLast);
			\draw [draw,-] (Neuron2Spike3DownLast) --  (Neuron2Spike4DownLast);
			\draw [draw,-] (Neuron2Spike4DownLast) --  (Neuron2Spike4UpLast);
			\draw [draw,-] (Neuron2Spike4DownLast) --  (Neuron2SpikeFinalLast);

			\node [input, right of= NeuronCircleLast, node distance=0.35cm] (Neuron1SpikeBeginDownLast_tot){};
			\node [input, above of= Neuron1SpikeBeginDownLast_tot, node distance=0.8cm] (Neuron1SpikeBeginLast_tot){};
			\node [input, above of= Neuron1SpikeBeginLast_tot, node distance=0.6cm] (Neuron1Spike1UpLast_tot){};
			\node [input, right of= Neuron1SpikeBeginLast_tot, node distance=0.1cm] (Neuron1Spike2DownLast_tot){};
			\node [input, left of= Neuron1SpikeBeginLast_tot, node distance=0.2cm] (Neuron1SpikeBeginLeftLast_tot){};
			\node [input, above of= Neuron1Spike2DownLast_tot, node distance=0.6cm] (Neuron1Spike2UpLast_tot){};
			\node [input, right of= Neuron1Spike2DownLast_tot, node distance=0.4cm] (Neuron1Spike3DownLast_tot){};
			\node [input, above of= Neuron1Spike3DownLast_tot, node distance=0.6cm] (Neuron1Spike3UpLast_tot){};
			\node [input, right of= Neuron1Spike3DownLast_tot, node distance=0.3cm] (Neuron1Spike4DownLast_tot){};
			\node [input, above of= Neuron1Spike4DownLast_tot, node distance=0.6cm] (Neuron1Spike4UpLast_tot){};
			\node [input, right of= Neuron1Spike3DownLast_tot, node distance=0.2cm] (Neuron1Spike4DownLast_tot){};
			\node [input, above of= Neuron1Spike4DownLast_tot, node distance=0.6cm] (Neuron1Spike4UpLast_tot){};
			\node [input, right of= Neuron1Spike4DownLast_tot, node distance=0.2cm] (Neuron1SpikeFinalLast_tot){};
			
			\draw [draw,-] (Neuron1SpikeBeginLast_tot) --  (Neuron1SpikeBeginLeftLast_tot);
			\draw [draw,-] (Neuron1SpikeBeginLast_tot) --  (Neuron1Spike1UpLast_tot);
			\draw [draw,-] (Neuron1SpikeBeginLast_tot) --  (Neuron1Spike2DownLast_tot);
			\draw [draw,-] (Neuron1Spike2DownLast_tot) --  (Neuron1Spike2UpLast_tot);
			\draw [draw,-] (Neuron1Spike2DownLast_tot) --  (Neuron1Spike3DownLast_tot);
			\draw [draw,-] (Neuron1Spike3DownLast_tot) --  (Neuron1Spike3UpLast_tot);
			\draw [draw,-] (Neuron1Spike3DownLast_tot) --  (Neuron1Spike4DownLast_tot);
			\draw [draw,-] (Neuron1Spike4DownLast_tot) --  (Neuron1Spike4UpLast_tot);
			\draw [draw,-] (Neuron1Spike4DownLast_tot) --  (Neuron1SpikeFinalLast_tot);
			
			\node [input, right of= NeuronCircleLast, node distance=0.3cm] (Neuron2SpikeBeginDownLast_tot){};
			\node [input, above of= Neuron2SpikeBeginDownLast_tot, node distance=0.8cm] (Neuron2SpikeBeginLast_tot){};
			\node [input, below of= Neuron2SpikeBeginLast_tot, node distance=0.45cm] (Neuron2Spike1UpLast_tot){};
			\node [input, right of= Neuron2SpikeBeginLast_tot, node distance=0.3cm] (Neuron2Spike2DownLast_tot){};
			\node [input, left of= Neuron2SpikeBeginLast_tot, node distance=0.1cm] (Neuron2SpikeBeginLeftLast_tot){};
			\node [input, below of= Neuron2Spike2DownLast_tot, node distance=0.45cm] (Neuron2Spike2UpLast_tot){};
			\node [input, right of= Neuron2Spike2DownLast_tot, node distance=0.2cm] (Neuron2Spike3DownLast_tot){};
			\node [input, below of= Neuron2Spike3DownLast_tot, node distance=0.45cm] (Neuron2Spike3UpLast_tot){};
			\node [input, right of= Neuron2Spike3DownLast_tot, node distance=0.2cm] (Neuron2Spike4DownLast_tot){};
			\node [input, below of= Neuron2Spike4DownLast_tot, node distance=0.45cm] (Neuron2Spike4UpLast_tot){};
			\node [input, right of= Neuron2Spike3DownLast_tot, node distance=0.3cm] (Neuron2Spike4DownLast_tot){};
			\node [input, below of= Neuron2Spike4DownLast_tot, node distance=0.45cm] (Neuron2Spike4UpLast_tot){};
			\node [input, right of= Neuron2Spike4DownLast_tot, node distance=0.2cm] (Neuron2SpikeFinalLast_tot){};
			
			\draw [draw,-] (Neuron2SpikeBeginLast_tot) --  (Neuron2Spike1UpLast_tot);
			\draw [draw,-] (Neuron2Spike2DownLast_tot) --  (Neuron2Spike2UpLast_tot);
			\draw [draw,-] (Neuron2Spike3DownLast_tot) --  (Neuron2Spike3UpLast_tot);
			\draw [draw,-] (Neuron2Spike4DownLast_tot) --  (Neuron2Spike4UpLast_tot);

            \draw [draw,-] (stateFeedback) -- node [pos=0.5]{} (stateFeedbackLast);
            \node [input, below of= stateFeedback, node distance=1.5cm] (stateFeedbackBelow){};
             \node [input, left of= stateFeedbackBelow, node distance=0.8cm] (stateFeedbackMultiple){};
            \draw [draw,-] (stateFeedbackMultiple) -- node [pos=0.5]{$y$} (stateFeedbackBelow);

             \node [input, below of= SpikySignal, node distance=1.5cm] (SpikySignalTot){};
              \node [input, right of= SpikySignalTot, node distance=0.8cm] (SpikySignalFinal){};

            \node [input, below of= SpikySignal, node distance=1.5cm] (SpikesCircle){};
            \draw [fill=white] (SpikesCircle) circle (0.1cm);
			
			\node [input, right of= SpikesCircle, node distance=0.1cm] (SpikesCircleRight){};
            \node [input, above of= SpikesCircle, node distance=0.1cm] (SpikesCircleAbove){};
              \node [input, below of= SpikesCircle, node distance=0.1cm] (SpikesCircleBelow){};
			
			 \draw [draw,->] (SpikySignal) -- node [pos=0.5]{} (SpikesCircleAbove);
         \draw [draw,->] (SpikySignalLast) -- node [pos=0.5]{} (SpikesCircleBelow);
           \draw [draw,->] (SpikesCircleRight) -- node [pos=0.5]{$y_{s}$} (SpikySignalFinal);
           \node at ($(Neuron2)!.45!(Neuron1Last)$) {\vdots};
		\end{tikzpicture}
	\end{center}
	\caption{Block diagram of the spike encoder. 
    }
	\label{Fig:blockDiagram_ASconvMultipleOutput}
\end{figure}
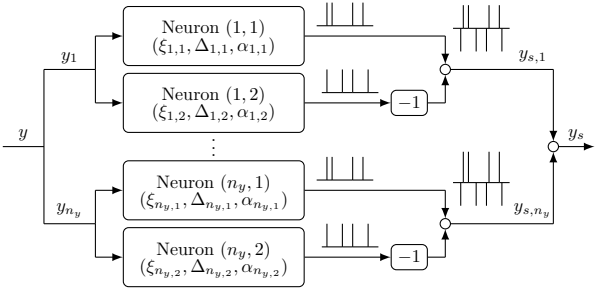

\noindent\emph{Neuronal dynamics.} 
Given the output component $y_i$, with  $i \in \{1,\dots, n_y\}$, we use  $\xi_{i,1}, \, \xi_{i,2} \in \R_{\geq0}$, to represent the membrane potential of the corresponding two neurons, which we denote neuron $(i,1)$ and neuron $(i,2)$, respectively. Variables $\xi_{i,1}$ and $\xi_{i,2}$ have the continuous-time dynamics 
\begin{eqn}
  \label{eq:thresholdFlow}
  	\dot{\xi}_{i,1} = \max\{0, y_i\}, & & 
	\dot{\xi}_{i,2} =  \max\{0, -y_i\}
\end{eqn}
between two successive spiking transmissions. 
Neurons are sensitive only to nonnegative inputs consistently with  \cite{gerstner2002spiking,petri2025emulation}. 
A spike is generated whenever one of the two neurons membrane potential is larger than or equal to a designed constant firing threshold, which we denote $\Delta_{i,\ell}>0$ with $\ell\in\{1,2\}$, namely when there exists $\ell\in\{1,2\}$ such that
\begin{equation}
\xi_{i,\ell} \geq \Delta_{i,\ell}.
\label{eq:triggeringRule}
\end{equation}
When (\ref{eq:triggeringRule}) holds, the corresponding membrane potential is reset to $0$, while the membrane potential of the other neuron is unchanged
\begin{eqn}
\xi_{i,\ell}^+=0, & & \xi_{i,3-\ell}^+= \xi_{i,3-\ell},
\end{eqn}
unless both membrane potentials satisfy (\ref{eq:triggeringRule}) in which case both are reset to $0$, i.e., $\xi_{i,\ell}^+=0$ and $\xi_{i,3-\ell}^+= 0$.

\smallskip

\noindent\emph{Spiking times.} To formalize the expression of $y_{s,i}$, we first need to define the sequence of spiking time instants generated by neurons $(i,1)$ and $(i,2)$. We have seen that a spike is generated whenever (\ref{eq:triggeringRule}) holds. This implies that each neuron $(i,\ell)$, $\ell\in\{1,2\}$, generates its own  sequence of spiking times $t^{(i,\ell)}_j$, $j\in\N_{\leq J_{(i,\ell)}}$ for some $J_{(i,\ell)}\in\No\cup\{\infty\}$, that is defined as
\begin{eqn}
t_{j+1}^{(i,\ell)}:= \inf\{t> t^{(i,\ell)}_j:\xi_{i,\ell}(t) \geq \Delta_{i,\ell} \}, & & t^{(i,\ell)}_0=0.
	\label{eq:triggeringRuleSpikingTimesEachNeuron}
\end{eqn}
We call $t^{(i,\ell)}_j$, $j\in\N_{\leq J_{(i,\ell)}}$, a sequence of spiking times with some slight abuse of terminology as its first element, namely $t_{0}^{(i,\ell)}$,  may not be a spiking time.  Indeed, when $\xi_{i,\ell}(0)<\Delta_{i,\ell}$, no spike occurs at time  $t^{(i,\ell)}_0 = 0$. On the other hand, when $\xi_{i,\ell}(0)\geq \Delta_{i,\ell}$, $0=t^{(i,\ell)}_{0}=t^{(i,\ell)}_{1}$, which means a spike occurs at the initial time $0$. 
The sequence of spiking times of the two-neuron network for each $i \in \{1,\dots, n_y\}$ is therefore given by 
$\{t^{\emph{(i)}}_j\}_{j\in \N_{\leq J_i}} := \{t^{(i,1)}_{j_1}\}_{j_1\in\N_{\leq J_{(i,1)}}} \cup \{t^{(i,2)}_{j_2}\}_{j_2 \in \N_{\leq  J_{(i,2)}}}$,
with $J_i\leq J_{(i,1)}+J_{(i,2)}$ as the sequences $\{t^{(i,1)}_{j_1}\}_{j_1\in\N_{\leq J_{(i,1)}}}$ and $\{t^{(i,2)}_{j_2}\}_{j_2 \in \N_{\leq  J_{(i,2)}}}$ may share common elements.
It follows that the sequence of spiking times for the overall network depicted in Fig. \ref{Fig:blockDiagram_ASconvMultipleOutput} is  given by $\{t_j\}_{j\in \N_{\leq J}} := \bigcup_{i=1}^{n_y}\{t^{\emph{(i)}}_{j_i}\}_{j_i\in \N_{\leq J_i}}$,
with $J\leq J_1+\ldots+J_{n_y}$, which we can equivalently write as
$t_{j+1} = \inf\big\{t> t_j: \exists (i,\ell)\in\{1,\ldots,n_y\}\times\{1,2\},\xi_{i,\ell}(t) \geq \Delta_{i,\ell} \big\}$ 
with $t_0=0$. 

\noindent\emph{Spiking signal $y_{s}$}. When $\xi_{i,\ell}$ with $\ell\in\{1,2\}$ is such that (\ref{eq:triggeringRule}) holds, a spike of fixed amplitude $\alpha_{i,\ell} \in \R_{> 0}$ is generated as output of the neuron $(i,\ell)$, where $\alpha_{i,\ell}$ is a design parameter. 
Like in  \cite{petri2024analysis, petri2025emulation}, the spikes generated by neuron $(i,2)$ are multiplied by a gain equal to $-1$. This allows to distinguish the spikes generated by neurons $(i,1)$ and $(i,2)$ that are sensitive to the positive and negative part of their signal input $y_i$, respectively. Consequently, $y_{s,i}$ is defined as
\begin{equation}
\begin{aligned}
	y_{s,i} & := \textstyle\sum_{j = 1}^{+\infty} \alpha_{i,1}\delta_{t^{(i,1)}_{j}} - \sum_{j = 1}^{+\infty} \alpha_{i,2}\delta_{t^{(i,2)}_{j}}.
    \end{aligned}
	\label{eq:SpikingMeasurement}
\end{equation}
Recall that $y_s=(y_{s,1},\ldots,y_{s,n_y})$ and  note that $y_s\in \mathcal{D}_{n_y}$. 

\begin{remark}\label{rem:SigmaDelta}
The spike encoder 
reminds of the continuous-time asynchronous $\Sigma-\Delta$ modulators, see, e.g.,  \cite{kikkert1975asynchronous,pavan2017understanding,lazar2004perfect,lazar2004time}. The link between $\Sigma-\Delta$ modulators and integrate-and-fire neurons with a refractory period is studied in \cite{lazar2004time}.  
 However, to the best of the authors' knowledge, formal and general methods for the analysis of feedback control loops with $\Sigma-\Delta$ modulators, as in this work, are lacking in the literature. 
\end{remark}

\subsection{Spike decoder}\label{subsect:spiking-decoder}

To convert the spiking signal  $y_s$ to the analog  signal $\widehat y$ used for control, we propose a spike decoder of the form of a LTI filter, namely  
\begin{eqn}
   \label{eq:filter}
   \dot{x}_\f = A_{\f}  x_{\f} + B_{\f} y_s, & & 
    \widehat y =  C_\f  x_\f,
\end{eqn}
where $x_\f \in \R^{n_\f}$ with $n_\f\in\Np$ is the filter state and $A_\f \in  \R^{n_\f \times n_\f}$, $B_\f \in \R^{n_\f \times n_y}$ and $C_\f \in \R^{n_y \times n_\f}$ are matrices to be designed. System (\ref{eq:filter}) admits a neuronal interpretation. When the matrices $A_\f$ and $B_\f$ are diagonal and $A_\f$ is Hurwitz, each state of (\ref{eq:filter}) filters the corresponding spike train input to produce an analog signal output, thereby reminding of a synaptic filter in neuronal models, see, e.g., \cite[Chapter 3.1]{gerstner2014neuronal}, \cite{roth2009modeling}.

Solutions to the system in \eqref{eq:filter} with spiking signal $y_s$ as input are well-defined using the convolution integral and the sifting property of Dirac pulses. In particular, 
\eqref{eq:filter} evolves according to \begin{eqn}
\dot{x}_\f = A_\f  x_\f
\end{eqn}
between two successive spiking instants. Moreover, at each spiking time $t_j$ for $j\in\No$, $x_\f$ experiences a jump modeled  by the discrete system
\begin{eqn}
x_\f^+ & = & x_\f + B_\f \Gamma(\xi)
\end{eqn}
with $\xi:=(\xi_{1,1},\xi_{1,2},\xi_{2,1},\xi_{2,2},\ldots,\xi_{n_y,1},\xi_{n_y,2})\in\R^{2n_y}$, 
\begin{eqn}\label{eq:Gamma}
\Gamma(\xi) := \left(\begin{smallmatrix}
\alpha_{1,1}\mathbf{1}_{\geq \Delta_{1,1}}(\xi_{1,1})-\alpha_{1,2}\mathbf{1}_{\geq \Delta_{1,2}}(\xi_{1,2}) \\
\vdots\\
\alpha_{n_y,1}\mathbf{1}_{\geq \Delta_{n_y,1}}(\xi_{n_y,1})-\alpha_{n_y,2}\mathbf{1}_{\geq \Delta_{n_y,2}}(\xi_{n_y,2})
\end{smallmatrix}\right).
\end{eqn}
Hence, when a single  neuron $(i,\ell)$ triggers a spike, with $(i,\ell)\in\{1,\ldots,n_y\}\times\{1,2\}$, the update of $x_\f$ writes $x_\f^+  =  x_\f + B_\f \big(0,\ldots, 0, (-1)^{1+\ell}\alpha_{i,\ell}, 0, \ldots, 0\big)$.

\begin{remark}
In Section \ref{subsect:spiking-encoder}, spikes of amplitude $\alpha_{i,\ell}$ are communicated when neuron $(i,\ell)$ reaches its threshold $\Delta_{i,\ell}$, with $i \in \{1,\dots, n_y\}$ and $\ell \in \{1,2\}$. An equivalent model consists in using only $1$-bit communications, where, for instance, one bit equal to $1$ (or $0$) is transmitted on channel $i \in \{1,\dots, n_y\}$, when neuron $(i, 1)$ (or $(i, 2)$) reaches its threshold.  The spike decoder then has to be modified to take into account the information coded in the received bit, i.e., 
$x_\f^+  =  x_\f + B_\f \big(0,\ldots, 0, (-1)^{1+\ell}\alpha_{i,\ell}, 0, \ldots, 0\big)$ where $\ell = 1$, when the received bit is equal to $1$, and $\ell = 2$, when the received bit is equal to $0$.
\end{remark}
\color{black}


\subsection{Controller}\label{Controller}
We consider  output-feedback dynamic controllers of the form 
\begin{eqn}
   \label{eq:controller}
     \dot x_\cont  =  f_\cont(x_\cont,\widehat y), & & 
     u = h_\cont(x_\cont, \widehat y),
\end{eqn}
where $x_c \in \R^{n_c}$ with $n_c\in\No$ is the controller state, $u\in \R^{n_u}$ is the control input to plant \eqref{eq:plant}, and $\widehat y \in \R^{n_y}$ is the input to the controller and corresponds to the output of the spike decoder in (\ref{eq:filter}). The design requirements on (\ref{eq:controller}) are presented  in  Section \ref{subsect:stability-property-Sigma-continuous} and discussed in Section~\ref{sect:Sigma-continuous-iss}. 

\subsection{Closed-loop system}\label{subsect:closed-loop-model}
The overall closed-loop model is given  by\footnote{Equation \eqref{eq:closed-loop} is written with some abuse of notation for space reasons: it uses one $\xi_{i,\ell}$ instead of the vector containing all $\xi_{i,\ell}$ with $i \in \{1,\dots, n_y\}$ and $\ell \in \{1,2\}$. Similar abuse of notation is used later in the paper.}  
%
\begin{eqn}\label{eq:closed-loop}
\left(\begin{matrix}\dot x_\p(t) \\ \dot x_\f(t) \\ \dot x_\cont(t) \\ \dot\xi_{i,\ell}(t)\end{matrix}\right) & = & \left(\begin{matrix} f_\p(x_\p(t),h_\cont(x_\cont(t),C_\f x_\f(t)),d(t))\\
A_\f x_\f(t)\\
f_\cont(x_\cont(t),C_\f x_\f(t)) \\
\max\{0,(-1)^{1+\ell}h_{\p,i}(x_\p(t),w(t))\}\end{matrix}\right) \\
& & \hfill t\in[t_j,t_{j+1}), \\ 
\left(\begin{matrix}x_\p(t_{j+1}^+) \\ x_\f(t_{j+1}^+) \\ x_\cont(t_{j+1}^+) \\ \xi_{i,\ell}(t_{j+1}^+)\end{matrix}\right) & = & \left(\begin{matrix} x_\p(t_{j+1})\\
x_\f(t_{j+1}) + B_\f\Gamma(\xi(t_{j+1}))\\
x_\cont(t_{j+1}) \\
\mathbf{1}_{\leq\Delta_{i,\ell}}\xi_{i,\ell}(t_{j+1})\end{matrix}\right) \\
t_{j+1} & = &  \inf\bigl\{t> t_j: \exists (i,\ell), \xi_{i,\ell}(t) \geq \Delta_{i,\ell} \bigr\}, 
\end{eqn}
with $t_0 = 0$ and $h_{\p,i}$ the $i^{\text{th}}$ component of the map $h_\p$, $i\in\{1,\ldots,n_y\}$. For the sake of convenience, we also write (\ref{eq:closed-loop}) for short as 
\begin{eqn}\label{eq:closed-loop-compact}
\dot\chi(t) & = & f(\chi(t),d(t),w(t)) \quad\quad t\in[t_j,t_{j+1})\\
\chi(t_{j+1}^{+}) & = &  g(\chi(t_{j+1})) & &\\
t_{j+1} & = &  \inf\bigl\{t> t_j: \exists (i,\ell), \xi_{i,\ell}(t) \geq \Delta_{i,\ell} \bigr\},
\end{eqn}
with $t_0 = 0$,  $\chi:=(x_\p,x_\f,x_\cont,\xi)\in\R^{n_\p}\times \R^{n_\f}\times\R^{n_\cont}\times\R^{2n_y}=\R^{n_\chi}$,  $n_\chi:=n_\p+n_\f+n_\cont+2n_y$ and the expressions of $f$ and $g$ follow from (\ref{eq:closed-loop}). By a solution to (\ref{eq:closed-loop-compact}), we mean, given any inputs $d\in\mathcal{L}_{\R^{n_d}}$ and $w\in\mathcal{L}_{\R^{n_w}}$, any function $\phi:[0,t^\star)\to\R^{n_\chi}$ with $t^{\star}\in\Rlp\cup\{\infty\}$ that  is absolutely continuous between any two successive jumps and that verifies $\dot\phi(t)=f(\phi(t),d(t),w(t))$ for almost all 
$t\in[t_j,t_{j+1})$, and $\phi(t_{j+1}^+)=g(\phi(t_{j}))$ for all $j\in\Np$. Also, when we consider a solution, we mean a maximal solution, i.e., a solution whose domain cannot be extended.

As mentioned in Section \ref{sect:problem-statement}, the objective is to ensure a stability property for the closed-loop system (\ref{eq:closed-loop}). For this purpose, we will design the spike amplitudes and the neuron's threshold such that $\alpha_{i,\ell}=\Delta_{i,\ell}$ for any $(i,\ell)\in\{1,\ldots,n_y\}\times\{1,2\}$; noting that the forthcoming results also apply when this is not the case, see  Remark~\ref{rem:NoUnitGain} below. 
To formally state the desired stability property for closed-loop system (\ref{eq:closed-loop}), we introduce the next two sets, given $\alpha:=(\alpha_{1,1},\alpha_{1,2},\alpha_{2,1},\ldots,\alpha_{n_y,2})\in\Rlp^{2n_y}$,
\begin{eqn}\label{eq:set-A}
\mathcal{X}_\alpha & := & \R^{n_\p+n_\f+n_\cont}  \times [0, \alpha_{1,1}] \times \cdots \times [0, \alpha_{n_y,2}]\\
\mathcal{A}_\alpha & := & \{0_{(n_\p+n_\f+n_\cont) \times 1}\}  \times [0, \alpha_{1,1}] \times \cdots \times [0, \alpha_{n_y,2}],
\end{eqn}
where $\mathcal{X}_\alpha$ is used to define the set of initial conditions for the solutions to the closed-loop system (\ref{eq:closed-loop}) and $\mathcal{A}_\alpha$ denotes the attractor set. We can now formalize the design objective. 

\begin{obj}\label{objective} Give design conditions on the spike decoder (\ref{eq:filter}) and controller (\ref{eq:controller}) such that the following holds when taking $\Delta_{i,\ell}=\alpha_{i,\ell}>0$ for any $(i,\ell)\in\{1,\ldots,n_y\}\times\{1,2\}$.
\begin{enumerate}
\item[(i)] Solutions to (\ref{eq:closed-loop}) initialized in $\mathcal{X}_\alpha$ with any inputs $d\in\mathcal{L}_{\R^{d}},w\in\mathcal{L}_{\R^{w}}$ are complete and their sequence of spiking times is in $\mathcal{T}_{\zf}^\infty$.
\item[(ii)] There exist $\beta\in\KL$ and $\gamma\in\Kinf$ such that  for any $\alpha_{i,\ell}>0$,  $(i,\ell)\in\{1,\ldots,n_y\}\times\{1,2\}$, for any solution $\phi$ to (\ref{eq:closed-loop}) initialized in $\mathcal{X}_{\alpha}$ with inputs $d\in\mathcal{L}_{\R^{n_d}}$ and $w\in\mathcal{L}_{\R^{n_w}}$, 
$|\phi(t)|_{\mathcal{A}_\alpha} \leq \max\big\{\beta(|\phi(0)|_{\mathcal{A}_\alpha},t),$ $\gamma(\|d\|_{[0,t]}),\gamma(\|w\|_{[0,t]}),\gamma(|\alpha|)\big\}$ for all $t\geq 0$, with $\mathcal{A}_\alpha$ in~(\ref{eq:set-A}). 
\end{enumerate}
\end{obj}

Objective \ref{objective} states a practical ISS property of the set $\mathcal{A}_\alpha$ in (\ref{eq:set-A}) for the  closed-loop system (\ref{eq:closed-loop}). This property implies that any solution to the closed-loop system (\ref{eq:closed-loop}) initialized in $\mathcal{X}_\alpha$ is complete, Zeno-free and its  components corresponding to $x_\p,x_\f,x_\cont$  converge to a neighborhood of the origin whose ``size'' depends on the norm of the disturbance input $d$ and of the measurement noise $w$, as in standard ISS, as well as  on the norm of the spike amplitude vector $\alpha$, which  can be made as small as desired by tuning the $\alpha_{i,\ell}$'s. Typically, the smaller $|\alpha|$, the more spikes are triggered as illustrated on an example in Section~\ref{subsect:example-simulations}.  The stability property in Objective \ref{objective}  is valid  as long as the membrane potentials $\xi_{i,\ell}$ are initialized in $[0,\Delta_{i,\ell}]$, which we can always enforce and is natural from a neuronal viewpoint.


To ensure Objective \ref{objective}, we do not directly work with system (\ref{eq:closed-loop}). Instead, we introduce auxiliary variables and use these to augment the closed-loop model (\ref{eq:closed-loop}) in the next section. 
We first study properties of the  augmented system, and then ensure that the solution space of this augmented system contains the solution space of the closed-loop system in \eqref{eq:closed-loop}.

\section{Augmenting the closed-loop model}\label{sect:loop-transformations}

We first define the spike-induced error (Section \ref{sec:SpikeInducedError}), 
after that we introduce two auxiliary variables to loosely speaking split the filter dynamics in (\ref{eq:filter}) into two systems (Section~\ref{subsect:splitting-filter}). Based on this, the augmented closed-loop model is derived (Section \ref{subsect:augmented-closed-loop}).

\subsection{Spike-induced error} \label{sec:SpikeInducedError}
We introduce the spike-induced error 
\begin{equation}
    \label{eq:spikeInducedError}
    e:= y - y_s \in \R^{n_y}.
\end{equation}
Variable $e$ is the mismatch between the analog plant output $y$ and its spiking counterpart $y_s$ generated by the spike encoder presented in Section \ref{subsect:spiking-encoder}. 
Given \eqref{eq:spikeInducedError}, the block diagram of the closed-loop system \eqref{eq:closed-loop}  shown in Fig. \ref{fig:closed-loop} can be transformed into the equivalent representation depicted in Fig. \ref{fig:LTIspikyComunication_firstTransformation}, where the spiking nature of the communicated output signal is incorporated in the spike-induced error $e$ only. In this way,  the closed-loop system can be modeled as a purely continuous-time system affected by spiking signal $e$. 
We exploit this observation in the sequel, in particular we model the closed-loop system as the interconnection of a continuous-time system and a spiking system, see Section \ref{sect:stability}. For that, we propose to rewrite the spike decoder  dynamics (\ref{eq:filter}).

\begin{figure}[t]
	\begin{center}
		\tikzstyle{blockB} = [draw, fill=blue!30, rectangle, rounded corners=2pt,
		minimum height=2em, minimum width=3em]  
		\tikzstyle{blockG} = [draw, fill=MyGreen!40, rectangle, rounded corners=2pt,
		minimum height=2em, minimum width=3em]
		\tikzstyle{blockR} = [draw, fill=red!40, rectangle, rounded corners=2pt,
		minimum height=2em, minimum width=3em]
		\tikzstyle{blockO} = [draw, rounded corners=2pt, minimum height=1.5em, fill=orange!20, minimum width=4em]
		\tikzstyle{input} = [coordinate]
		\tikzstyle{blockW} = [draw, rounded corners=2pt, minimum height=1.5em, fill=white!20, minimum width=2.5em]
		\tikzstyle{input1} = [coordinate]
		\tikzstyle{blockCircle} = [draw, circle]
		\tikzstyle{sum} = [draw, circle, minimum size=.3cm]
		\tikzstyle{blockSensor} = [draw, rounded corners=2pt, fill=white!20, draw= blue!80, line width= 0.8mm, minimum height=10em, minimum width=13em]
		
		\begin{tikzpicture}[auto, node distance=2cm,>=latex , scale=0.75,transform shape] 
			
			\node [input, name=stateFeedbackSpiky] {};
            \node [blockW, right of=stateFeedbackSpiky, node distance=3.3cm] (SAconverter) { 
				$\begin{array}{c}
					\text{Spike decoder}
				\end{array}$
			};
			
            \node [input, right of=stateFeedbackSpiky, node distance=1cm] (CircleError){};
            \draw [fill=white] (CircleError) circle (0.3cm);
			
			\node [input, right of= CircleError, node distance=0.3cm] (CircleErrorRight){};
			
			\node [input, left of= CircleError, node distance=0.3cm] (CircleErrorLeft){};
            \node [input, above of= CircleError, node distance=0.3cm] (CircleErrorAbove){};
            \draw [draw, ->] (stateFeedbackSpiky) -- node [pos=0.5]{$y$} (CircleErrorLeft);
            \node [input, above of= CircleError, node distance=1cm] (EmulationError){};
			\draw [draw, ->] (EmulationError) -- node [pos=0.5]{$-e$} (CircleErrorAbove);
            \draw [draw, ->] (CircleErrorRight) -- node [pos=0.5]{} (SAconverter);

			\node [blockW, anchor=west] (Controller) at ([xshift=0.7cm]SAconverter.east) { 
				$\begin{array}{c}
					\text{Controller}
				\end{array}$
			};
			
			\draw [draw,->] (SAconverter) -- node [pos=0.5]{$\widehat{y}$} (Controller);
			
			\node [blockW, anchor=west] (system) at ([xshift=0.7cm]Controller.east) { 
				$\begin{array}{c}
					\text{Plant}
				\end{array}$
			};

            			\node [input, above of= system, node distance=1cm] (PlantDisturbance){};
\draw [draw, ->] (PlantDisturbance) -- node [pos=0.5, right]{$(d,w)$} (system);

			\draw [draw,->] (Controller) -- node [pos=0.5]{$u$} (system);
			

			\node [input, right of= system, node distance=1.6cm] (SystemOutputSpiky){};
			\node [input, right of= system, node distance=1.1cm] (SystemFeedbackSpiky){};
			\node [input, below of= SystemFeedbackSpiky, node distance=0.8cm] (SystemFeedbackSpikyBelow){};
            \node [input, left of= SystemFeedbackSpikyBelow, node distance=3.8cm] (NetworkIn){};
            \node [input, left of= NetworkIn, node distance=3cm] (NetworkOut){};
			\node [input, below of= stateFeedbackSpiky, node distance=0.8cm] (stateFeedbackSpikyBelow){};
			
			\draw [draw, ->] (system) -- node [pos=0.5]{$y$} (SystemOutputSpiky);
			\draw [draw, -] (SystemFeedbackSpiky) --  (SystemFeedbackSpikyBelow);
            \draw [draw, -] (stateFeedbackSpiky) --  (stateFeedbackSpikyBelow);
            \draw [draw,-] (SystemFeedbackSpikyBelow) --  (stateFeedbackSpikyBelow);
            \node [input, right of= Controller, node distance=1.2cm] (ControllerRight){};
            
		\end{tikzpicture}
	\end{center}
	\caption{Closed-loop with spike-induced error.}
	\label{fig:LTIspikyComunication_firstTransformation}
\end{figure}
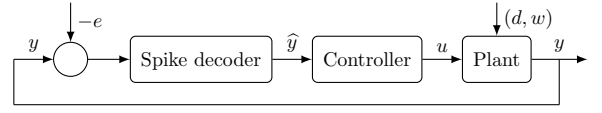

\subsection{Splitting the spike decoder dynamics}\label{subsect:splitting-filter}
As $y_s = y -e \in \R^{n_y}$,  \eqref{eq:filter} becomes
\begin{eqn}
   \label{eq:filter_with_e}
   \dot{x}_\f = A_\f  x_\f + B_\f (y -e), & & 
    \widehat y =  C_\f  x_\f.
\end{eqn}

We introduce two auxiliary variables denoted $x_{\f,\n}\in\R^{n_\f}$ and $x_{\f,e}\in\R^{n_\f}$, respectively. The idea is to assign dynamics and initial conditions to $x_{\f,\n}$ and $x_{\f,e}$ so that $x_\f=x_{\f,\n}-x_{\f,e}$ and $x_{\f,\n}$ exhibits continuous, non-spiking dynamics, while the spiking dynamics is captured in $x_{\f,e}$. Before giving their dynamics and explaining their meaning, we  highlight that these two variables do \textit{not} need to be  implemented: these are only introduced for the sake of design and analysis. We define the dynamics of  $x_{\f,\n}$ as
\begin{eqn}
   \label{eq:filter-x-f-n}
 \dot{x}_{\f,\n} = A_\f x_{\f,\n} + B_\f y, & & 
\widehat y_{\n} =  C_\f x_{\f,\n}.
\end{eqn}
System (\ref{eq:filter-x-f-n}) has the same dynamics as in (\ref{eq:filter}) except that it is the analog, non-spiking signal $y$ that is being filtered and not the spiking one $y_s$. This is why we use the index subscript ``$\n$'' for ``nominal''. 
We similarly define the dynamics of $x_{\f,e}$ as
\begin{eqn}
   \label{eq:filter-x-f-e}
 \dot{x}_{\f,e} =  A_\f x_{\f,e} + B_\f e, & & 
\widehat y_{e}  =   C_\f x_{\f,e}.
\end{eqn}
This time the spike-induced error $e$ is being filtered in (\ref{eq:filter-x-f-e}), thereby justifying the index subscript ``$e$''. Given the spiking nature of $e$,  recalling that $e = y-y_s$, with $y_s$ defined in \eqref{eq:SpikingMeasurement}, we derive, like in Section~\ref{subsect:spiking-decoder}, that between two successive spiking instants
\begin{eqn}
\dot{x}_{\f,e} & = & A_\f x_{\f,e} + B_\f y 
\end{eqn}
and at each spiking instant, with $\Gamma$ in (\ref{eq:Gamma}), 
\begin{eqn}
x_{\f,e}^+ & = & x_{\f,e} - B_\f \Gamma(\xi).
\end{eqn}

The next lemma  establishes a relation between solutions to  the spike decoder (\ref{eq:filter}) and solutions to the nominal filter (\ref{eq:filter-x-f-n}) and to the spike-induced error filter (\ref{eq:filter-x-f-e}). 

\begin{lemma}\label{lem:relation-between-filters-solutions} Let  $y_s\in \mathcal{D}_{n_y}$ with  spiking times $\{t_{j}\}_{j\in \N_{\leq J}}$ with $J\in\No\cup\{\infty\}$ and  $y\in\mathcal{L}_{\R^{n_y}}$. Any solution  $\phi_\f$ to (\ref{eq:filter}) with input $y_s$  verifies  $\phi_\f=\phi_{\f,\n}-\phi_{\f,e}$ on $[0,t_J)$, where $\phi_{\f,\n}$ is any solution to (\ref{eq:filter-x-f-n}) with input $y$  and $\phi_{\f,e}$ is the solution to (\ref{eq:filter-x-f-e}) with input $e=y-y_s$ and initial condition $\phi_{\f,e}(0)=\phi_{\f,\n}(0)-\phi_{\f}(0)$. 
\end{lemma}

\begin{proof} Let $y_s\in \mathcal{D}_{n_y}$ and $y\in\mathcal{L}_{\R^{n_y}}$. Consider an arbitrary solution $\phi_\f$ to (\ref{eq:filter}) with input $y_s$, and an arbitrary solution $\phi_{\f,\n}$ to (\ref{eq:filter-x-f-n}) with input $y$. Let $\phi_{\f,e}$ be the solution to (\ref{eq:filter-x-f-e}) as specified in Lemma \ref{lem:relation-between-filters-solutions}.  
The solutions $\phi_\f$,  $\phi_{\f,\n}$ and $\phi_{\f,e}$  are defined on $[0,t_J)$, as the corresponding dynamics are LTI and their inputs are well-defined on  the interval $[0,t_J)$. We  have that for all $t\in[0,t_1)$,  $\dot{\phi}_{\f,\n}(t)-\dot{\phi}_{\f,e}(t)=A_{\f}\phi_{\f,\n}(t)+B_\f y(t) - A_{\f}\phi_{\f,e}(t)-B_\f y(t)=A_\f(\phi_{\f,\n}(t)-\phi_{\f,e}(t))$. Since  $\phi_{\f,\n}(0)-\phi_{\f,e}(0)=\phi_{\f}(0)$ and by (\ref{eq:filter}), we deduce that $\phi_{\f}=\phi_{\f,\n}-\phi_{\f,e}$ on $[0,t_1)$. At $t_1$, a spike occurs for both $\phi_\f$ and $\phi_{\f,e}$ as their inputs, namely $y_s$ and $y-y_s$ exhibit spikes at the same time instants. Consequently, $\phi_{\f,\n}(t_1^+)-\phi_{\f,e}(t_1^+)=\phi_{\f,\n}(t_1)-\phi_{\f,e}(t_1)+B_\f a_1=\phi_\f(t_1)+B_\f a_1$ where $a_1$ denotes the spike amplitude of $y_s$ at time $t_1$. On the other hand, $\phi_\f(t_1^+)=\phi_\f(t_1)+B_\f a_1$ and thus $\phi_{\f,\n}(t_1^+)-\phi_{\f,e}(t_1^+)=\phi_\f(t_1^+)$. By induction, We derive  $\phi_{\f,\n}-\phi_{\f,e}=\phi_\f$ on $[0,t_J)$.
\end{proof}

\subsection{Augmented closed-loop system}\label{subsect:augmented-closed-loop}

Given (\ref{eq:filter-x-f-n}) and (\ref{eq:filter-x-f-e}), we augment the closed-loop system (\ref{eq:closed-loop}) into the  system of larger dimension presented in (\ref{eq:closed-loop-augmented}).
\begin{table*}[t]
{\small\begin{equation}\label{eq:closed-loop-augmented}
\begin{array}{rllllllllllllllll}
\left(\begin{matrix}\dot x_\p(t) \\ \dot x_{\f,\n}(t) \\ \dot x_{\f,e}(t) \\ \dot x_\cont(t) \\ \dot\xi_{i,\ell}(t)\end{matrix}\right) \! & \! =\! & \!\!\left(\begin{matrix} f_\p(x_\p(t),h_\cont(x_\cont(t),C_\f (x_{\f,\n}(t)-x_{\f,e}(t))),d(t))\\
A_\f x_{\f,\n}(t)+B_\f h_\p(x_\p(t),w(t))\\
A_\f x_{\f,e}(t)+B_\f h_\p(x_\p(t),w(t))\\
f_\cont(x_\cont(t),C_\f (x_{\f,\n}(t)-x_{\f,e}(t)) \\
\max\{0,(-1)^{1+\ell}h_{\p,i}(x_\p(t),w(t))\}\end{matrix}\right) 
\, t\in[t_j,t_{j+1}), \,\,
\left(\begin{matrix}x_\p(t_{j+1}^+) \\ x_{\f,\n}(t_{j+1}^+)  \\ x_{\f,e}(t_{j+1}^+)\\ x_\cont(t_{j+1}^+) \\ \xi_{i,\ell}(t_{j+1}^+)\end{matrix}\right) \!&\! =\! &\! \left(\begin{matrix} x_\p(t_{j+1})\\
x_{\f,\n}(t_{j+1}) \\
\!\!x_{\f,e}(t_{j+1}) \!-\! B_\f\Gamma(\xi(t_{j+1}))\!\!\\
x_\cont(t_{j+1}) \\
\mathbf{1}_{\leq\Delta_{i,\ell}}\xi_{i,\ell}(t_{j+1})\!\!\end{matrix}\right)  \\
t_{j+1}\! &\! = & \! \inf\bigl\{t> t_j: \exists (i,\ell)\in\{1,\ldots,n_y\}\times\{1,2\}, \xi_{i,\ell}(t) \geq \Delta_{i,\ell} \bigr\}, \quad
t_0  = 0.
\end{array}
\end{equation}}
\noindent\rule{\textwidth}{0.4pt}
\end{table*}

The next result formalizes the relation between the solutions to (\ref{eq:closed-loop}) and (\ref{eq:closed-loop-augmented}). 

\begin{theorem}
    \label{thm:solutions-closed-loop-models}
Given any $d\in\mathcal{L}_{\R^{n_d}}$ and $w\in\mathcal{L}_{\R^{n_w}}$, consider any solution $\phi:=(\phi_\p,\phi_\f,\phi_\cont,\xi)$ to (\ref{eq:closed-loop}) with inputs $d$ and $w$ and $\phi^{\aug}=(\phi^{\aug}_\p,\phi^{\aug}_{\f,\n},\phi^{\aug}_{\f,e},\phi^{\aug}_\cont,\phi^{\aug}_\xi)$ solution to (\ref{eq:closed-loop-augmented}) with initial condition $\phi^{\aug}_p(0)=\phi_\p(0)$, $\phi^{\aug}_{\f,\n}(0)-\phi^{\aug}_{\f,e}(0)=\phi_{\f}(0)$, $\phi^{\aug}_\cont(0)=\phi_\cont(0)$ and $\phi_\xi^{\aug}(0)=\phi_\xi(0)$ and the same inputs $d$ and $w$. Then $\phi$ and $\phi^{\aug}$ have the same domain, 
the same sequence of spiking times $\{t_{j}\}_{j\in \N_{\leq J}}$ with $J\in\No\cup\{\infty\}$ and   $\phi^{\aug}_p(t)=\phi_\p(t)$, $\phi^{\aug}_{\f,\n}(t)-\phi^{\aug}_{\f,e}(t)=\phi_{\f}(t)$, $\phi^{\aug}_\cont(t)=\phi_\cont(t)$ and $\phi_\xi^{\aug}(t)=\phi_\xi(t)$ for any $t$ in the domain of the solutions.
\end{theorem}

\begin{proof} Let $d\in\mathcal{L}_{\R^{n_d}}$ and $w\in\mathcal{L}_{\R^{n_w}}$, $\phi:=(\phi_\p,\phi_\f,\phi_\cont,\xi)$ be a solution to (\ref{eq:closed-loop}) with inputs $d$ and $w$,  and  sequence of spiking times $\{t_j\}_{j\in\N_{\leq J}}\in\mathcal{T}$ with $J\in\No\cup\{\infty\}$. Let  $\phi^{\aug}=(\phi^{\aug}_\p,\phi^{\aug}_{\f,\n},\phi^{\aug}_{\f,e},\phi^{\aug}_\cont,\phi_{\xi}^{\aug})$  be the solution to (\ref{eq:closed-loop-augmented}) as specified in Theorem  \ref{thm:solutions-closed-loop-models}. We denote the sequence of spiking times of $\phi^\aug$, $\{t^\aug_j\}_{j\in\N_{\leq J^\aug}}\in\mathcal{T}$ with $J^\aug\in\No\cup\{\infty\}$. Given the initial condition of $\phi^\aug$, the fact that both $\phi$ and $\phi^\aug$ are subject to the same inputs $d$ and $w$ and the definitions of systems (\ref{eq:closed-loop}) and (\ref{eq:closed-loop-augmented}),  we derive that $\phi^{\aug}_p=\phi_\p$, $\phi^{\aug}_{\f,\n}-\phi^{\aug}_{\f,e}=\phi_{\f}$, $\phi^{\aug}_\cont=\phi_\cont$ and $\phi_\xi^{\aug}=\phi_\xi$ on $[0,\min\{t_1,t_1^\aug\})$. The fact that $\phi_\xi^{\aug}=\phi_\xi$ on $[0,\min\{t_1,t_1^\aug\})$ implies that $t_1=t_1^\aug$. This implies that  $\phi_j(t_1^+)=\phi_j(t_1)=\phi_j^\aug(t_1)=\phi_j^\aug(t_1^+)$ for $j\in\{\p,\cont\}$. On the other hand, $\phi^\aug_{\f,\n}(t_1^+)-\phi^\aug_{\f,e}(t_1^+)=\phi^\aug_{\f,\n}(t_1)-\phi^\aug_{\f,e}(t_1)+B_\f \Gamma(\phi^\aug_\xi(t_1))=\phi_\f(t_1)+B_\f \Gamma(\phi^\aug_\xi(t_1))=\phi_\f(t_1^+)$ in view of (\ref{eq:closed-loop}).  By induction, we derive that $\phi$ and $\phi^\aug$ are defined on the same domain, that they have the same spiking time sequences and the relation on their values established on $[0,t_1]$ above applies on the whole domain.
\end{proof}


\section{Closed-loop system stability}\label{sect:stability}

In this section, we address Objective \ref{objective}. 
For this purpose, we first concentrate on the augmented system (\ref{eq:closed-loop-augmented}). We interpret system (\ref{eq:closed-loop-augmented}) as the feedback interconnection of two subsystems (Section \ref{subsect:sys-interconnection}), similarly to what is done in other hybrid contexts, see, e.g.,  \cite{Nesic_Teel_netw_TAC04,maass-et-al-tac-eyes-event}. Stability properties for each of these subsystems are then either established or assumed (Sections \ref{subsect:stability-property-Sigma-s} and \ref{subsect:stability-property-Sigma-continuous}). Afterwards, we  derive  properties for the augmented closed-loop system (\ref{eq:closed-loop-augmented}), and  we  show how these translate to the fulfillment of Objective \ref{objective} for the original closed-loop  system (\ref{eq:closed-loop}) by leveraging Theorem  \ref{thm:solutions-closed-loop-models} (Section~\ref{subsect:stability-properties-closed-loop}).

\subsection{The augmented closed-loop model as a feedback interconnection}\label{subsect:sys-interconnection}

We interpret the augmented closed-loop model  (\ref{eq:closed-loop-augmented}) as the feedback interconnection of the $(x_\p,x_{\f,\n},x_\cont)$-system, called $\Sigma_{\continuous}$, with the $(x_{\f,e},\xi)$-system, denoted $\Sigma_s$, see Fig. \ref{fig:feedback-interconnection}. System $\Sigma_\continuous$ exhibits only continuous-time dynamics, thereby justifying the index subscript ``cont'' for ``continuous'', and corresponds to the closed-loop dynamics in absence of the spike-induced error, namely 
\vspace{-0.8cm}
\begingroup
\small\begin{eqn}\label{eq:sys-Sigma-continuous}
\Sigma_\continuous:\left\{\begin{array}{rllll}
\!\left(\begin{matrix}\!\dot x_\p \!\\\! \dot x_{\f,\n}\!  \\\! \dot x_\cont\! \end{matrix}\right) \!&\! = \!&\! \left(\begin{matrix} \!f_\p(x_\p,h_\cont(x_\cont,C_\f x_{\f,\n}\!-\! \widehat{y}_{e}),d)\!\\
A_\f x_{\f,\n}+B_\f h_\p(x_\p,w)\\
f_\cont(x_\cont,C_\f x_{\f,\n}-\widehat{y}_{e}) \end{matrix}\right)\\
\!y \!& = &\! h_\p(x_\p,w).
\end{array}\right.
\end{eqn}
\endgroup
\noindent The inputs to $\Sigma_\continuous$ are $d$, $w$ and $\widehat{y}_e=C_\f x_{\f,e}$, while the output is $y$. On the other hand, system $\Sigma_s$ exhibits spikes, which justifies the use of index subscript ``$s$'' for ``spike''. System $\Sigma_s$ is given by 
\vspace{-0.8cm}
\begingroup
\small\begin{eqn}\label{eq:sys-Sigma-s}
\Sigma_s & : & \left\{\begin{array}{rllllll}\left(\begin{matrix} \dot x_{\f,e}(t) \\  \dot\xi_{i,\ell}(t)\end{matrix}\right) & = & \left(\begin{matrix}
A_\f x_{\f,e}(t)+B_\f y(t)\\
\max\{0,(-1)^{1+\ell}y_i(t)\}\end{matrix}\right)\\ 
& & \hfill t\in[t_j,t_{j+1})\\ 
\left(\begin{matrix} x_{\f,e}(t_{j+1}^+)\\  \xi_{i,\ell}(t_{j+1}^+)\end{matrix}\right) & = & \left(\begin{matrix} x_{\f,e}(t_{j+1}) - B_\f\Gamma(\xi(t_{j+1}))\\
\mathbf{1}_{\leq\Delta_{i,\ell}}\xi_{i,\ell}(t_{j+1})\end{matrix}\right)  \\
\widehat{y}_e & = & C_\f x_{\f,e}\\
t_{j+1} & = &  \inf\bigl\{t> t_j: \exists (i,\ell), \xi_{i,\ell}(t) \geq \Delta_{i,\ell} \bigr\}\\ 
t_0 & = & 0.
\end{array}\right.
\end{eqn}
\endgroup
The input  to $\Sigma_s$ is $y$ and the output is $\widehat{y}_e$. We explain next how to select the neuron parameters $\alpha_{i,\ell}$ and $\Delta_{i,\ell}$ of the spike encoder in Section \ref{subsect:spiking-encoder} and the parameters of the  spike decoder in Section \ref{subsect:spiking-decoder} so that system $\Sigma_s$ exhibits a  practical stability property.

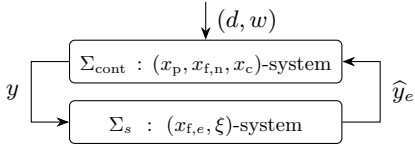
\begin{figure}
\begin{center}
\begin{tikzpicture}[
  block/.style = {draw, rectangle, rounded corners=2pt,
                  minimum width=2cm, minimum height=0.5cm,
                  text width=4cm, align=center, font=\footnotesize, scale = 0.85},
  >=Stealth]
 
\node[block] (top) at (0,0.8) {$\Sigma_\continuous\,:\,(x_{\p},x_{\f,\n},x_\cont)$-system};
\node[block] (bot) at (0,0)   {$\Sigma_s\,:\,(x_{\f,e},\xi)$-system};
 
\coordinate (rTop) at ($(top.east)+(0.5,-0)$);
\coordinate (rBot) at ($(bot.east)+(0.5,0)$);
\draw[->] (bot.east) -- (rBot) -- (rTop) -- (top.east);
\node[font=\footnotesize] at ($(rTop)!0.5!(rBot)$) [right=1pt] {$\widehat y_{e}$};
 
\coordinate (lTop) at ($(top.west)+(-0.5,0)$);
\coordinate (lBot) at ($(bot.west)+(-0.5,0)$);
\draw[->] (top.west) -- (lTop) -- (lBot) -- (bot.west);
\node[font=\footnotesize] at ($(lTop)!0.5!(lBot)$) [left=1pt] {$y$};
\coordinate (dTop) at ($(top.north)+(0,0.5)$);
\draw[->] (dTop) -- (top.north);
\node[font=\footnotesize] at ($(dTop)!0.5!(top.north)$) [right=0.5pt] {$(d,w)$};
 \end{tikzpicture}
\end{center}
\caption{Augmented closed-loop system (\ref{eq:closed-loop-augmented}) as the feedback interconnection of $\Sigma_\continuous$ in (\ref{eq:sys-Sigma-continuous}) with $\Sigma_s$ in (\ref{eq:sys-Sigma-s}).}\label{fig:feedback-interconnection}
\end{figure}

\subsection{Stability property of $\Sigma_s$}\label{subsect:stability-property-Sigma-s}

We establish a practical stability property for system $\Sigma_s$ in \eqref{eq:sys-Sigma-s} in the next proposition. 
For this purpose we  define the sets, given $\alpha=(\alpha_{1,1},\alpha_{1,2},\alpha_{2,1},\ldots,\alpha_{n_y,2})\in\R^{2n_y}$, 
\begin{eqn}\label{eq:As_set}
\mathcal{X}_{\alpha}^{s} & := &  \R^{n_\f}\times [0, \alpha_{1,1}]\times \dots \times [0, \alpha_{n_y,2}]\\
    \mathcal{A}_{\alpha}^{s} & := &  \{0_{n_\f\times1}\} \times [0, \alpha_{1,1}]\times \dots \times [0, \alpha_{n_y,2}].
\end{eqn}
Set $\mathcal{X}_\alpha^s$ is used to define the considered set of initial conditions for the solutions to $\Sigma_s$ and $\mathcal{A}_\alpha^s$ denotes the attractor set in the next proposition.

\begin{propos}\label{thm:practical-iss-Sigma-s} 
Consider system \eqref{eq:sys-Sigma-s} with $A_\f$ Hurwitz and $\alpha_{i,\ell} = \Delta_{i,\ell}$ for all $(i,\ell) \in \{1, \dots, n_y\}\times \{1,2\}$. The following holds.
\begin{enumerate}[label=(\roman*)]
\item Given any $y \in \mathcal{L}_{\R^{n_y}}$, all solutions initialized in $\mathcal{X}_{\alpha}^{s}$ in (\ref{eq:As_set}) with input $y$, are complete and their sequences of spiking times belong to $\mathcal{T}_{\zf}^\infty$.
\item There exist $\beta_s \in \exp$-$\KL$ and $\tilde{\gamma}_s\in\Rlp$ such that for all $y \in \mathcal{L}_{\R^{n_y}}$, all solutions $\phi_s$   initialized  in $\mathcal{X}_{\alpha}^{s}$ with input $y$  satisfy
$|\phi_{s}(t)|_{\mathcal{A}_{\alpha}^{s}} \leq \max\{ \beta_s(|\phi_{s}(0)|_{\mathcal{A}_{\alpha}^{s}},t), \tilde{\gamma}_s|\alpha|\}$
for all $t\geq 0$, with $\mathcal{A}_\alpha^s$ in (\ref{eq:As_set}).
\end{enumerate}
\vspace{-0.3cm}
 
\end{propos}
The proof of Proposition \ref{thm:practical-iss-Sigma-s} is given in Appendix \ref{Appendix:ProofThmBoundOnSpikyFilterOutput}. In terms of design requirements, Proposition \ref{thm:practical-iss-Sigma-s} only asks for $\alpha_{i,\ell}=\Delta_{i,\ell}$ for all $(i,\ell)\in\{1,\ldots,n_y\}\times\{1,2\}$, as already commented (see also Remark \ref{rem:NoUnitGain} below) and $A_\f$ to be Hurwitz.  Proposition \ref{thm:practical-iss-Sigma-s}\emph{(i)} establishes that the neuronal architecture does not exhibit Zeno; it can actually be shown that there exists a semi-global dwell-time for each sequence of spikes for each neuron, when the infinity norm of $y$ is bounded, which is the case under the conditions of Theorem \ref{thm:practical-iss-closed-loop} in Section~\ref{subsect:stability-properties-closed-loop}. 
Proposition \ref{thm:practical-iss-Sigma-s}\textit{(ii)} guarantees that  
a practical stability property for system (\ref{eq:sys-Sigma-s}). 
The upper-bound in the inequality in  Proposition~\ref{thm:practical-iss-Sigma-s}\emph{(ii)} is independent of  input $y$, which is why it is possible to invoke cascade arguments in the sequel.

\begin{remark}\label{rem:NoUnitGain}
    {In Proposition \ref{thm:practical-iss-Sigma-s}, and in the following results, we select $\alpha_{i,\ell} = \Delta_{i,\ell}$ for all $i \in \{1, \dots, n_y\}$ and $\ell \in \{1,2\}$ to simplify notations. Indeed, if $\alpha_{i,\ell} \neq \Delta_{i,\ell}$ and $\frac{\alpha_{i,1}}{\Delta_{i,1}} = \frac{\alpha_{i,2}}{\Delta_{i,2}}$, then the spike encoder, in addition to convert the continuous signal $y_{i}$ into the spiking signal $y_{s,i}$, it will scale it with a gain given by $\frac{\alpha_{i,\ell}}{\Delta_{i,\ell}}$, see \cite[Theorem~1]{petri2025emulation}; the results in this paper follow \emph{mutatis mutandi}.}
\end{remark}

\subsection{Stability property of $\Sigma_\continuous$}\label{subsect:stability-property-Sigma-continuous}

We assume that controller (\ref{eq:controller}) and the spike decoder  in (\ref{eq:filter}) are designed to ensure the next property.

\begin{ass}\label{ass:iss-Sigma-continuous} System $\Sigma_\continuous$ is input-to-state stable with respect to inputs $d$, $w$ and $\widehat{y}_{e}$. 
\end{ass}

Assumption \ref{ass:iss-Sigma-continuous} means that the controller in (\ref{eq:controller}) is designed to ensure an ISS property for the closed-loop system it forms with the plant model (\ref{eq:plant}) and the nominal filter in (\ref{eq:filter-x-f-n}), i.e., when the noisy analog output $y$ is \emph{continuously} communicated to the filter. We elaborate on Assumption \ref{ass:iss-Sigma-continuous} in Section \ref{sect:Sigma-continuous-iss}.


\subsection{Stability properties for the closed-loop models}\label{subsect:stability-properties-closed-loop}
We introduce the next two sets, for given $\alpha\in\R^{2n_y}$,
\begin{eqn}\label{eq:X-aug-A-aug}
\mathcal{X}_{\alpha}^{\aug} & := &  \R^{n_\p +2n_\f+n_\cont}\times  [0, \alpha_{1,1}]\times \dots \times [0, \alpha_{n_y,2}]\\
    \mathcal{A}_{\alpha}^{\aug} & := &  \{0_{(n_\p +2n_\f+n_\cont)\times 1}\} \times [0, \alpha_{1,1}]\times \dots \times [0, \alpha_{n_y,2}].
\end{eqn}
Like in (\ref{eq:set-A}) and in Section \ref{subsect:stability-property-Sigma-s}, $\mathcal{X}_\alpha^\aug$ is used to define the considered set of initial conditions for the solutions to the augmented closed-loop system in (\ref{eq:closed-loop-augmented}) and $\mathcal{A}_\alpha^\aug$ defines the attractor of interest. Given Proposition \ref{thm:practical-iss-Sigma-s} and Assumption \ref{ass:iss-Sigma-continuous}, we derive the next properties for the augmented closed-loop system in (\ref{eq:closed-loop-augmented}). The proof is given in Appendix \ref{appendix:proof-prop-stability-augmented-closed-loop}.

\begin{propos}
\label{prop:stability-augmented-closed-loop} Consider the augmented closed-loop system in (\ref{eq:closed-loop-augmented}) with $A_\f$ Hurwitz, $\alpha_{i,\ell}=\Delta_{i,\ell}>0$ for all $(i,\ell)\in\{1,\ldots,n_y\}\times\{1,2\}$ and suppose Assumption \ref{ass:iss-Sigma-continuous} is satisfied. Then the following holds.
\begin{enumerate}[label=(\roman*)]
\item Given any $d\in\mathcal{L}_{\R^{n_d}}$ and $w\in\mathcal{L}_{\R^{n_w}}$, any solution initialized  in $\mathcal{X}_{\alpha}^{\aug}$ in (\ref{eq:X-aug-A-aug}) with  $\phi^\aug_{\f,e}(0)=0$ and inputs $d$ and $w$, is complete and its sequence of spiking times belongs to $\mathcal{T}_{\zf}^\infty$.
\item There exist $\beta_\aug\in\KL$ and $\gamma_\aug\in\Kinf$ such that for any  $d\in\mathcal{L}_{\R^{n_d}}$ and $w\in\mathcal{L}_{\R^{n_w}}$, any  solution $\phi^\aug=(\phi^\aug_\p,\phi^\aug_{\f,\n},\phi^\aug_{\f,e},\phi^\aug_\cont,\phi^\aug_\xi)$ initialized in $\mathcal{X}_{\alpha}^{\aug}$ with $\phi^\aug_{\f,e}(0)=0$ and inputs $d$ and $w$ verifies 
$|\phi^\aug(t)|_{\mathcal{A}_{\alpha}^{\aug}} \leq \max\big\{\beta_\aug(|\phi^\aug(0)|_{\mathcal{A}_{\alpha}^{\aug}},t), $ $
\gamma_\aug(\|d\|_{[0,t]}),\gamma_\aug(\|w\|_{[0,t]}),\gamma_\aug(|\alpha|)\big\}$
for all $t\geq 0$, with $\mathcal{A}^{\aug}_{\alpha}$ in (\ref{eq:X-aug-A-aug}).
\end{enumerate}
\vspace{-0.3cm}
\end{propos}


Proposition \ref{prop:stability-augmented-closed-loop}\emph{(ii)} ensures that solutions to the augmented closed-loop system in (\ref{eq:closed-loop-augmented}) are complete and Zeno free and that the system exhibits a practical stability-like property. 
We talk of a stability-like property because it only holds for a set of initial conditions of Lebesgue measure zero as $\phi_{\f,e}$ has to be equal to $0$, thereby contradicting the rationale behind Lyapunov stability. Nevertheless, we recall that the variable $x_{\f,e}$ is only introduced for analysis purpose. Moreover,  we show in the next result that the desired  practical ISS property for the actual closed-loop model (\ref{eq:closed-loop}) can be ensured and more generally that Objective \ref{objective} is fulfilled under Assumption \ref{ass:iss-Sigma-continuous}. 

\begin{theorem}\label{thm:practical-iss-closed-loop} Consider the closed-loop system in (\ref{eq:closed-loop}) with $A_\f$ Hurwitz, $\alpha_{i,\ell}=\Delta_{i,\ell}$ for any $(i,\ell)\in\{1,\ldots,\ell\}\times\{1,2\}$ and suppose Assumption \ref{ass:iss-Sigma-continuous} is satisfied. Then Objective \ref{objective}(i)-(ii) hold. 
\end{theorem}

\begin{proof} Let $d\in\mathcal{L}_{\R^{n_d}}$, $w\in\mathcal{L}_{\R^{n_w}}$ and $\phi$ be a solution to system (\ref{eq:closed-loop}) with inputs $d$ and $w$ initialized in $\mathcal{X}_\alpha$ in (\ref{eq:set-A}). Let $\phi^\aug=(\phi^{\aug}_{\p},\phi^{\aug}_{\f,\n},\phi^{\aug}_{\f,e},\phi^{\aug}_{\cont}, \phi^{\aug}_{\xi})$ be the solution to the augmented closed-loop model (\ref{eq:closed-loop-augmented}) with the same inputs $d$ and $w$, $\phi^\aug_\p(0)=\phi_\p(0)$, $\phi^{\aug}_{\f,\n}(0)=\phi_\f(0)$, $\phi^{\aug}_{\f,e}(0)=0$, $\phi^{\aug}_{\cont}(0)=\phi_\cont(0)$ and $\phi^{\aug}_{\xi}(0)=\phi_\xi(0)$. As $\phi^\aug_{\f,e}(0)=\phi^\aug_{\f,\n}(0)-\phi_{\f}(0)=0$, we derive from Theorem~\ref{thm:solutions-closed-loop-models} that $\phi$ and $\phi^\aug$  have the same domain and the same sequence of spiking times. This implies that $\phi$ is complete and that its sequence of spiking times belongs to $\mathcal{T}_{\zf}^\infty$ by Proposition \ref{prop:stability-augmented-closed-loop}\emph{(i)}: Objective \ref{objective}\emph{(i)} holds. 

Let $t\geq 0$. By Theorem \ref{thm:solutions-closed-loop-models}\emph{(ii)},  $\phi^{\aug}_\p(t)=\phi_\p(t)$, $\phi^{\aug}_{\f,\n}(t)-\phi^{\aug}_{\f,e}(t)=\phi_{\f}(t)$, $\phi^{\aug}_\cont(t)=\phi_\cont(t)$ and $\phi_\xi^{\aug}(t)=\phi_\xi(t)$.
Hence, noting that $\phi_\xi(t)\in[0,\alpha_{1,1}]\times\ldots\times[0,\alpha_{n_y,2}]$,  we have that 
$|\phi(t)|_{\mathcal{A}_\alpha} = |(\phi_\p(t),\phi_{\f}(t),\phi_\cont(t))|
=   |(\phi^\aug_\p(t),\phi^\aug_{\f,\n}(t)-\phi^\aug_{\f,e}(t),\phi^\aug_\cont(t))|$. 
This implies
\begin{eqn}\label{eq:proof-thm-practical-iss-bound-distance-to-attractor}
|\phi(t)|_{\mathcal{A}_\alpha} & \leq &  |(\phi^\aug_\p(t),\phi^\aug_{\f,\n}(t),\phi^\aug_\cont(t))|+|\phi^\aug_{\f,e}(t)|.
\end{eqn}
In the following, we derive upper-bounds of the two terms in the right-hand side of (\ref{eq:proof-thm-practical-iss-bound-distance-to-attractor}). 
By Proposition~\ref{thm:practical-iss-Sigma-s}(\emph{ii}), 
$|\phi_{s}(t)|_{\mathcal{A}_{\alpha}^{s}} \leq \max\{ \beta_s(|\phi_{s}(0)|_{\mathcal{A}_{\alpha}^{s}},t),$ $ \tilde{\gamma}_s|\alpha|\}$, 
with $\phi_s=(\phi^\aug_{\f,e},\phi^\aug_\xi)$. We have  $|\phi_s(t)|_{\mathcal{A}_{\alpha}^{s}}=|\phi^\aug_{\f,e}(t)|$ with $\mathcal{A}_\alpha^s$ in (\ref{eq:As_set}), as $\phi^\aug_\xi(t)\in[0,\alpha_{1,1}]\times\ldots\times[0,\alpha_{n_y,2}]$. Hence,
$|\phi_{\f,e}^\aug(t)| \leq \max\{ \beta_s(|\phi_{\f,e}^{\aug}(0)|,t), \tilde\gamma_s|\alpha|\}$ 
and, as $\phi_{\f,e}^\aug(0)=0$, 
\begin{eqn}\label{eq:proof-thm-closed-loop-stability-bound-phi-f-e}
|\phi_{\f,e}^\aug(t)| \leq \tilde{\gamma}_s|\alpha|.
\end{eqn}
On the other hand, 
$|(\phi^\aug_\p(t),\phi^\aug_{\f,\n}(t),\phi^\aug_\cont(t))| \leq |\phi^\aug(t)|_{\mathcal{A}_\alpha^\aug}$, 
with $\mathcal{A}_\alpha^\aug$ in (\ref{eq:X-aug-A-aug}). Hence, by Proposition~\ref{prop:stability-augmented-closed-loop}\emph{(ii)}, 
\begin{eqn}\label{eq:proof-thm-closed-loop-stability-bound-aug-variables}
|(\phi^\aug_\p(t),\phi^\aug_{\f,\n}(t),\phi^\aug_\cont(t))| \! & \! \leq \! & \!\max\big\{\beta_\aug(|\phi^\aug(0)|_{\mathcal{A}_{\alpha}^{\aug}},t),\\
& & \hspace{-2cm} \gamma_\aug(\|d\|_{[0,t]}),\gamma_\aug(\|w\|_{[0,t]}),\gamma_\aug(|\alpha|)\big\}.
\end{eqn}
By (\ref{eq:proof-thm-practical-iss-bound-distance-to-attractor}), (\ref{eq:proof-thm-closed-loop-stability-bound-phi-f-e}) and (\ref{eq:proof-thm-closed-loop-stability-bound-aug-variables}), we deduce that $
|\phi(t)|_{\mathcal{A}_\alpha} \leq \max\big\{\beta_\aug(|\phi^\aug(0)|_{\mathcal{A}_{\alpha}^{\aug}},t),
\gamma_\aug(\|d\|_{[0,t]}),\gamma_\aug(\|w\|_{[0,t]}),$ $\gamma_\aug(|\alpha|)\big\} + \tilde\gamma_s|\alpha|$. Noting that $|\phi^\aug(0)|_{\mathcal{A}_{\alpha}^{\aug}}=|\phi(0)|_{\mathcal{A}_{\alpha}}$ given the initial conditions of $\phi^\aug$, we obtain the desired inequality in Objective \ref{objective}\emph{(ii)} with $\beta=2\beta_\aug\in\KL$ and $\gamma=2\max\{\gamma_\aug,\tilde{\gamma}_s\id\}\in\Kinf$. 
\end{proof}


\begin{remark}\label{rem:compact-attractor} In Assumption \ref{ass:iss-Sigma-continuous}, the attractor is the origin $(x_\p,x_{\f,\n},x_\cont)=0$. When the attractor is a generic compact attractor $\mathcal{A}_\continuous\subset\R^{n_\p+n_\f+n_\cont}$ instead\footnote{Namely, when  there exist $\beta_\continuous\in\KL$ and $\gamma_\continuous\in\Kinf$ such that for any  $d\in\mathcal{L}_{\R^{n_d}}$, $w\in\mathcal{L}_{\R^{n_w}}$ and  $\widehat{y}_e\in\mathcal{L}_{\R^{n_y}}$, any solution $\phi_\continuous$ to $\Sigma_\continuous$ with inputs $d$, $w$ and $\widehat{y}_e$ verifies $|\phi_\continuous(t)|_{\mathcal{A_\continuous}}\leq\max\{\beta_\continuous(|\phi_\continuous(0)|_{\mathcal{A_\continuous}},t),\gamma_\continuous(\|\widehat{y}_e\|_{[0,t]}),\gamma_\continuous(\|d\|_{[0,t]}),$ $\gamma_\continuous(\|w\|_{[0,t]})\}$ for all $t\geq 0$.}, Theorem \ref{thm:practical-iss-closed-loop} still applies by replacing the definition of set $\mathcal{A}_\alpha$ with $\mathcal{A}_\continuous\times  [0, \alpha_{1,1}] \times \cdots \times [0, \alpha_{n_y,2}]$, as the proofs of Proposition \ref{prop:stability-augmented-closed-loop} and Theorem \ref{thm:practical-iss-closed-loop} carry over. 
\end{remark}

\begin{remark}\label{rem:spikingCommunInContToPlantChannel}
We consider in this work spiking communications between the sensor and the controller, see Fig. \ref{Fig:blockDiagram_ASconvMultipleOutput}. We can apply the proposed design methodology in the same spirit to neuron-spiking communication in other channels as well, e.g., between the controller and actuator, but also both sensor-to-controller channel and controller-to-actuator channel, or even multi-channel distributed configurations.  
 For instance, when we consider the case where the neuron-spiking communications are in the controller-to-actuator channel, the spiking signal is $u_s$, which is the spiking version of the analog control input $u$ generated by the controller. Thus, the spike encoder consists of $2n_u$ neurons, and the spike decoder filters $u_s$. 
Following similar steps as in Sections \ref{sect:design}-\ref{sect:stability}, we can define the spike-induced error $ e = u-u_s \in \R^{n_u}$, and augment the new closed-loop model similarly to Section \ref{sect:loop-transformations}. The resulting new system $\Sigma_s$ has input $u$ and output $\widehat{u}_e$, and following similar steps as in Section \ref{subsect:stability-property-Sigma-s}, a practical stability property can be ensured. On the other hand, the new system $\Sigma_ {\text{cont}}$ would have inputs $d$, $w$, and $\widehat{u}_e = C_f x_{f,e}$, and output $u$, and similarly to Assumption \ref{ass:iss-Sigma-continuous}, we assume that we can design the controller and the spike decoder such that an ISS property holds. The results in Section \ref{subsect:stability-properties-closed-loop} hold \emph{mutatis mutandi}, and a practical ISS property can be proven for the closed-loop system with neuron-spiking communications in the controller-to-actuator channel.

\end{remark}

\begin{remark}\label{rem:robustness-to-noise} As shown in Theorem \ref{thm:practical-iss-closed-loop}\emph{(i)}, the proposed scheme does not exhibit  Zeno phenomenon despite measurement noise. No knowledge on the measurement noise is required for the design of the communication scheme. This is in stark contrast  with  event-triggered control, in which measurement noise requires either specific parameter tuning or time-regularization methods, and thus  clocks, to be Zeno free, see \cite{scheres2024robustifying} and the references therein.
\end{remark}
\section{Satisfaction of Assumption \ref{ass:iss-Sigma-continuous}}\label{sect:Sigma-continuous-iss}

While there exist various tools to design input-to-state stabilizing controllers for classes of nonlinear systems, see, e.g., \cite{mironchenko2023-iss-book,freeman-kokotovic-book,praly-bresch-pietri-book2-2022}, the design problem posed by Assumption \ref{ass:iss-Sigma-continuous} is non-standard. Indeed,  it involves designing both the controller and the filter dynamics of the spike decoder to ensure an ISS property for the obtained closed-loop model with plant (\ref{eq:plant}). Nevertheless, we may still exploit existing design techniques to ensure Assumption \ref{ass:iss-Sigma-continuous}. One approach consists in separately designing the controller (\ref{eq:controller}) and the spike decoder (\ref{eq:filter}) and then to derive conditions under which Assumption \ref{ass:iss-Sigma-continuous} holds. This is the approach followed in Section \ref{subsect:assumption-Sigma-continuous-nonlinear} for a class of nonlinear systems. An alternative approach consists in designing first the spike decoder  (\ref{eq:filter}) and then seek for a suitable controller (\ref{eq:controller}) that input-to-state stabilize the plant with the spike decoder. We follow this approach both in Section \ref{subsect:LTI}, where we show that Assumption \ref{ass:iss-Sigma-continuous} can always be satisfied for any stabilizable and detectable LTI plant models, and in Section~\ref{sect:example} on a nonlinear example.

\subsection{A class of nonlinear systems}\label{subsect:assumption-Sigma-continuous-nonlinear}

We consider the case where the plant output map $h_p$ is linear, i.e., $h_p(x_\p,w)=C_\p x_\p+w$ for some $C_\p\in\R^{n_y\times n_\p}$; $n_w=n_y$ here.  We design the spike decoder  in  (\ref{eq:filter}) such that $n_\f=n_y$ and $C_\f=I_{n_y}$; further conditions are specified below. 
Let the mismatch between the noise-free plant output, $C_\p x_\p$, and the nominal filtered version of the noisy output, namely $\widehat{y}_{\n}$, i.e., the output of \eqref{eq:filter-x-f-n}, be denoted
\begin{eqn}\label{eq:tilde-y}
\widetilde y & := & C_\p x_\p-\widehat{y}_{\n}\in\R^{n_y}.
\end{eqn}
Along the solutions to $\Sigma_\continuous$ in (\ref{eq:sys-Sigma-continuous}), we have 
\begin{eqn}\label{eq:sys-widetilde-y}
\dot{\widetilde{y}} & = & q(\widetilde{y},x_{\p},x_{\cont},x_{\f,\n},\widehat{y}_e,d,w),
\end{eqn}
where $q(\widetilde{y},x_{\p},x_{\cont},x_{\f,\n},\widehat{y}_e,d) := C_p f_\p(x_\p,h_\cont(x_\cont,C_\p x_\p-\widetilde y-\widehat{y}_e),d)-A_\f x_{\f,\n}-B_\f (C_\p x_\p+w)$. The next result provides conditions on the controller and on system (\ref{eq:sys-widetilde-y}) under which  Assumption \ref{ass:iss-Sigma-continuous} holds. 


\begin{propos}\label{prop:conditions-for-assumption-on-Sigma-continuous} Let $n_\f=n_y$, $C_\f=I_{n_y}$ and suppose the following holds.
\begin{enumerate}[label=(\roman*)]
\item There exists $C_\p\in\R^{n_y\times n_\p}$ such that $h_p(x_\p,w)=C_\p x_\p+w$ for any $x_\p\in\R^{n_\p}$ and $w\in\R^{n_y}$.
\item The system
$(\dot x_{\p},\dot x_{\cont})  =  (f_{\p}(x_\p,h_\cont(x_\cont,C_\p x_\p+w),d), 
f_\cont(x_\cont,C_p x_\p+w))$ is ISS with respect to inputs $d$ and $w$ with gains $\gamma_{d}$ and $\gamma_{w}$, respectively.
\item System (\ref{eq:sys-widetilde-y}) with $n_\f=n_y$ and $C_\f=I_{n_y}$ is ISS with respect to inputs $(x_\p,x_{\cont})$, $\widehat{y}_e$, $d$ and $w$ with gains $\theta_{x}$, $\theta_e$ , $\theta_d$, and $\theta_w$.
\item $\theta_{x}\circ\gamma_{w}(2 s)<s$ for any $s>0$.
\end{enumerate}
\vspace{-0.4cm}Then Assumption \ref{ass:iss-Sigma-continuous} holds.
\end{propos}

\begin{proof} As $C_\f=I_{n_y}$, $C_\f x_{\f,\n} = \widehat{y}_\n$ in (\ref{eq:filter-x-f-n}). We thus have  $C_\f x_{\f,\n}-\widehat{y}_e=\widehat{y}_{\n} - \widehat{y}_{e}=C_\p x_\p-C_\p x_\p+\widehat{y}_{\n} - \widehat{y}_{e}=C_\p x_\p-\widetilde{y} - \widehat{y}_e$. Consequently, the $(x_\p,x_\cont)$-system in $\Sigma_\continuous$ in (\ref{eq:sys-Sigma-continuous}) becomes
\begin{eqn}\label{eq:proof-prop-sys-cont-bis}
\left(\begin{smallmatrix}\dot x_\p \\  \dot x_\cont \end{smallmatrix}\right) & = & \left(\begin{smallmatrix} f_\p(x_\p,h_\cont(x_\cont,C_\p x_\p-\widetilde{y} - \widehat{y}_e),d)\\
f_\cont(x_\cont,C_\p x_\p -\widetilde{y} - \widehat{y}_e) \end{smallmatrix}\right).
\end{eqn}
Proposition \ref{prop:conditions-for-assumption-on-Sigma-continuous}\emph{(ii)} implies that this  system is ISS with respect to inputs $\widetilde{y}+\widehat{y}_e$ and $d$ with gains $\gamma_{w}$ and  $\gamma_{d}$. As $\gamma_w\in\Kinf$,  using \cite[Eq. (8)]{kellett2014compendium}, we have for any $s_1,s_2\geq 0$, $\gamma_{w}(s_1+s_2)\leq \max\{\gamma_{w}(2s_1),\gamma_{w}(2s_2)\}$. We derive that system (\ref{eq:proof-prop-sys-cont-bis}) is ISS with respect to inputs $\widetilde{y}$, $\widehat{y}_e$ and $d$ with gains $\gamma_{w}(2\id)$, $\gamma_{w}(2\id)$ and $\gamma_{d}$.  By Proposition \ref{prop:conditions-for-assumption-on-Sigma-continuous}\emph{(iv)}, it follows from the application of the small-gain theorem in \cite[Theorem 3.119]{praly-bresch-pietri-book1-2022} that the interconnected system (\ref{eq:sys-widetilde-y}), (\ref{eq:proof-prop-sys-cont-bis}) is ISS with respect to inputs $\widehat{y}_e$, $d$ and $w$.  To conclude the proof, we observe that system (\ref{eq:sys-widetilde-y}),  (\ref{eq:proof-prop-sys-cont-bis}) corresponds to $\Sigma_\continuous$ in (\ref{eq:sys-Sigma-continuous}) by changing the coordinates from $(x_\p,x_\cont,\widetilde y)$ to $(x_\p,x_{\f,\n},x_\cont)$ using a linear, bijective map. 
Consequently, $\Sigma_\continuous$ is ISS with respect to inputs $\widehat{y}_e$, $d$ and $w$: Assumption \ref{ass:iss-Sigma-continuous} holds. \end{proof}

Proposition \ref{prop:conditions-for-assumption-on-Sigma-continuous} 
requires the controller to be input-to-state stabilizing for the plant model (\ref{eq:plant}) where the inputs are the disturbance $d$ acting on the plant dynamics and additive perturbations on the measurement. Techniques ensuring this property can be found for classes of systems in e.g., \cite{freeman-kokotovic-book,praly-bresch-pietri-book2-2022}. Proposition \ref{prop:conditions-for-assumption-on-Sigma-continuous}\emph{(iii)} on the other hand requires the filter dynamics to be designed such that (\ref{eq:sys-widetilde-y}) satisfies an ISS property. We provide below conditions under which this property holds. Finally, a small gain condition is imposed in Proposition \ref{prop:conditions-for-assumption-on-Sigma-continuous}\emph{(iv)} under which the  ISS properties ensured by the controller and the spike decoder ensure the satisfaction of Assumption \ref{ass:iss-Sigma-continuous}. 

\begin{lemma}\label{lem:sys-widetilde-y-iss} Suppose the following  holds.
\begin{enumerate}[label=(\roman*)]
\item There exists $C_\p\in\R^{n_y\times n_\p}$ such that $h_p(x_\p,w)=C_\p x_\p+w$ for any $x_\p\in\R^{n_\p}$ and $w\in\R^{n_y}$. 
\item There exists $\gamma_{\p}\in\Kinf$ such that $|C_\p f_{\p}(x_{\p},h_{\cont}(x_{\cont},C_{\p} x_{\p}$ $-\tilde y-\widehat y_e),d)|\leq \max\{\gamma_{\p}(|(x_{\p},x_{\cont})|),\gamma_{\p}(|d|)\}$ for any $x_\p\in\R^{n_\p}$, $x_\cont\in\R^{n_\cont}$, $d\in\R^{n_d}$, $\tilde y,\widehat{y}_e\in\R^{n_w}$.
\item $n_\f=n_y$, $A_\f$ is Hurwitz, $B_\f=-A_\f$ and $C_\f=I_{n_y}$ in (\ref{eq:filter}). 
\end{enumerate}
\vspace{-0.4cm}Then system (\ref{eq:sys-widetilde-y}) is ISS with respect to inputs $(x_\p,x_\cont)$, $\widehat{y}_e$,  $d$, and $w$ with  gains $c\gamma_\p$, $0$, $c \gamma_\p$ and $c \max\{\gamma_\p,\id\}$ for some constant $c>0$. 
\end{lemma}

\begin{proof} As $B_\f= -A_\f$ by Lemma \ref{lem:sys-widetilde-y-iss}\emph{(iii)}, system (\ref{eq:sys-widetilde-y}) becomes using (\ref{eq:tilde-y}) 
\begin{eqn}\label{eq:proof-lem-sys-widetilde-y}
\dot{\widetilde{y}} & = & C_\p f_{\p}(x_{\p},h_{\cont}(x_{\cont},C_{\p} x_{\p}-\widetilde y -\widehat{y}_e),d)\\
& & - A_\f \widehat{y}_n + A_\f (C_\p x_\p+w)\\
& = & A_\f \widetilde y +  C_\p f_{\p}(x_{\p},h_{\cont}(x_{\cont},C_{\p} x_{\p}-\widetilde y -\widehat{y}_e),d) +A_\f  w.
\end{eqn}
As $A_\f$ is Hurwitz by Lemma \ref{lem:sys-widetilde-y-iss}\emph{(iii)} and by invoking Lemma \ref{lem:sys-widetilde-y-iss}\emph{(ii)}, we deduce that system (\ref{eq:proof-lem-sys-widetilde-y}) is ISS with respect to inputs respect to  inputs $(x_\p,x_\cont)$, $\widehat y_e$, $d$ and $w$ with  gains $c\gamma_\p$, $0$, $c \gamma_\p$ and $c \max\{\gamma_\p,\id\}$ for some constant $c>0$.
\end{proof}

Lemma \ref{lem:sys-widetilde-y-iss} provides conditions on the plant  (\ref{eq:plant}), the controller (\ref{eq:controller}) and the spike decoder (\ref{eq:filter}) under which Proposition \ref{prop:conditions-for-assumption-on-Sigma-continuous}\emph{(iii)} holds. In particular, Lemma \ref{lem:sys-widetilde-y-iss}\emph{(ii)} is a bounding condition on  $\vartheta:(x_\p,x_\cont,d,\tilde y,\widehat{y}_e,)\mapsto C_\p f_{\p}(x_{\p},h_{\cont}(x_{\cont},C_{\p} x_{\p}-\tilde y-\widehat y_e),d)$. This condition is for instance satisfied when the map $\vartheta$ is independent of $\tilde y$, zero at zero and uniformly continuous. 
We highlight that Proposition \ref{prop:conditions-for-assumption-on-Sigma-continuous} and the subsequent Lemma~\ref{lem:sys-widetilde-y-iss} only provide cases where Assumption \ref{ass:iss-Sigma-continuous} holds for nonlinear systems; others can be envisioned as exemplified in Section \ref{sect:example}.

\subsection{LTI plant models}\label{subsect:LTI}
We now consider the case where system \eqref{eq:plant} is LTI
\begin{eqn}
   \label{eq:plant-lti}
     \dot x_\p  =  A_\p x_\p + B_\p u + F_\p d,  & &     y & = & C_\p x_\p+w, 
\end{eqn}
where $A_\p\in\R^{n_\p\times n_\p}$, $B_\p\in\R^{n_\p\times n_u}$, $F_\p\in\R^{n_\p\times n_d}$, $C_\p\in\R^{n_\p\times n_y}$. We make the next assumption on  system \eqref{eq:plant-lti}.

\begin{ass}\label{ass:stab-detect-lti} 
The pair $(A_\p,B_\p)$ is stabilizable and the pair $(A_\p,C_\p)$ is detectable.   
\end{ass}


System (\ref{eq:plant-lti}) with the nominal filter (\ref{eq:filter-x-f-n}) leads to 
\begin{equation}
\label{eq:sys-lti-plant-nominal-filter}
\begin{aligned}
\left(\begin{smallmatrix}\dot x_\p \\ \dot x_{\f,\n} \end{smallmatrix}\right)  &=  A_\continuous 
\left(\begin{smallmatrix}x_\p \\ x_{\f,\n} \end{smallmatrix}\right) + B_\continuous u + F_\continuous d + E_\continuous w\\
\widehat{y}_\n  &=  C_\continuous \left(\begin{smallmatrix}x_\p \\ x_{\f,\n} \end{smallmatrix}\right) 
\end{aligned}
\end{equation}
with $
A_\continuous  := \left(\begin{smallmatrix} A_\p & 0 \\
B_\f C_\p & A_\f    
\end{smallmatrix}\right)$, $  
B_\continuous := \big(B_\p^\top \,\,\,0\big)^\top$,  $F_\continuous:=\big(F_\p^\top \,\,\, 0\big)^\top$,  $E_\continuous:=\big(0\,\,\, B_\f^\top\big)^\top$ and $C_\continuous := \big(0 \,\,\, C_\f\big)$. 
The goal is to design a linear (dynamic) output-feedback controller such that the origin of the obtained closed-loop system with  (\ref{eq:sys-lti-plant-nominal-filter}) with $d\equiv 0$ and $w\equiv 0$ is globally exponentially stable.  Indeed, this would imply that the corresponding system  $\Sigma_\continuous$ in (\ref{eq:sys-Sigma-continuous}) is  ISS with respect to inputs $\widehat{y}_e$, $d$ and $w$ by \cite[Chapter 4.9]{khalil2002nonlinear} thereby ensuring the satisfaction of Assumption \ref{ass:iss-Sigma-continuous}. This design problem admits a solution if and only if the pair $(A_\continuous,B_\continuous)$  is stabilizable and the pair $(A_\continuous,C_\continuous)$ is detectable. The next lemma  states that this is the case under Assumption~\ref{ass:stab-detect-lti}, by suitably designing the spike decoder dynamics. 

\begin{lemma}\label{lem:stab-detect-with-filter} Suppose  Assumption \ref{ass:stab-detect-lti} holds, select $A_\f$, $B_\f$ and $C_\f$ as follows.
\begin{enumerate}[label=(\roman*)]
\item $A_\f$ is Hurwitz.
\item $B_\f$ has full column rank.
\item $\left(\begin{smallmatrix}
A_\f - \lambda I_{n_\f} & B_\f\\
C_\f  & 0
\end{smallmatrix}\right)$ has full column rank for all $\lambda \in \C$ with $\Re(\lambda) \geq0$.
\end{enumerate}
\vspace{-0.4cm}Then the pair $(A_\continuous,B_\continuous)$  is stabilizable and the pair $(A_\continuous,C_\continuous)$ is detectable.
\end{lemma}

\begin{proof} According to Popov-Belevitch-Hautus test \cite[Theorem 14.2]{hespanha2009linear}, the pair $(A_\continuous,B_\continuous)$ is stabilizable if and only if $\rank(A_\continuous-\lambda I\,\, B_\continuous)=n_\p+n_\f$ for all $\lambda\in\C$ with $\Re(\lambda)\geq 0$. Let $\lambda\in \C$ with $\Re(\lambda)\geq 0$,
\begin{eqn}
(A_\continuous-\lambda I\,\, B_\continuous) & := & \left(\begin{smallmatrix}
A_\p - \lambda I_{n_\p} & 0 & B_\p\\
B_\f C_\p & A_\f - \lambda I_{n_\f}  & 0
\end{smallmatrix}\right).
\end{eqn}
The rank of $\big(\begin{matrix}
A_\p - \lambda I_{n_\p} &  B_\p
\end{matrix}\big)$ is $n_\p$ as $(A_\p,B_\p)$ is stabilizable by Assumption \ref{ass:stab-detect-lti}. Moreover the rank of $A_\f - \lambda I_{n_\f}$ is $n_\f$ as $A_\f$ is Hurwitz. Consequently, $(A_\continuous-\lambda I\,\, B_\continuous) $ is of rank $n_\p+n_\f$: $(A_\continuous,B_\continuous)$ is stabilizable.

Similarly, by Popov-Belevitch-Hautus test \cite[Theorem~16.5]{hespanha2009linear}, the pair $(A_\continuous,C_\continuous)$ is detectable if and only if for any $\lambda\in\C$ with $\Re(\lambda)\geq 0$,
\begin{eqn}\label{eq:proof-lem-detect}
\rank\left(\begin{smallmatrix}
A_\p - \lambda I_{n_\p} & 0 \\
B_\f C_\p & A_\f -\lambda I_{n_\f} \\
0 & C_\f
\end{smallmatrix}\right) & = & n_\p+n_\f.
\end{eqn}
Condition \eqref{eq:proof-lem-detect} is equivalent to  
\begin{eqn}\label{eq:proof-lem-detect2}
\left(\begin{smallmatrix}
A_\p - \lambda I_{n_\p} & 0 \\
B_\f C_\p & A_\f -\lambda I_{n_\f} \\
0 & C_\f
\end{smallmatrix}\right) \left(\begin{smallmatrix} v_1 \\ v_2 \end{smallmatrix}\right)   = 0 & & \Rightarrow (v_1,v_2)=(0,0).
\end{eqn}
The right-hand side above yields 
$\left(\begin{smallmatrix}
A_\f - \lambda I_{n_\f} & B_\f\\
C_\f  & 0
\end{smallmatrix}\right) \left(\begin{smallmatrix} v_2 \\ C_\p v_1 \end{smallmatrix}\right)  = 0$.
By  Lemma \ref{lem:stab-detect-with-filter}\emph{(iii)}, using $\Re(\lambda)\geq 0$, we obtain 
$v_2 = 0$ and $C_\p v_1 = 0$. As the right-hand side of \eqref{eq:proof-lem-detect2} also gives $(A_p-\lambda I_{n_p})v_1=0$, from detectability of  $(A_\p, C_\p)$, we have $v_1 = 0$. Thus \eqref{eq:proof-lem-detect2} holds and $(A_\continuous,C_\continuous)$ detectable. 
\end{proof}

The conditions of Lemma \ref{lem:stab-detect-with-filter} can always be ensured as we are free to design $n_\f$, $A_\f,B_\f$, and $C_\f$ as desired. Note that Lemma \ref{lem:stab-detect-with-filter}\emph{(iii)}   is related to $(A_\f, B_\f, C_\f)$ being minimum phase. When these conditions hold, we can always design controller (\ref{eq:controller}) of the form 
\begin{eqn}\label{eq:controller-lti}
\dot x_\cont = A_\cont x_\cont + B_\cont \widehat y, & & u = C_\cont x_\cont, 
\end{eqn}
with $A_\cont\in\R^{n_\cont\times n_\cont}$, $B_\cont\in\R^{n_\cont\times n_\f}$, $C_\cont\in\R^{n_u\times n_\cont}$ and $n_\cont=n_\p+n_\f$ so that Assumption \ref{ass:iss-Sigma-continuous} holds as stated below.

\begin{propos}\label{prop:lti} Suppose the following holds.
\begin{enumerate}[label=(\roman*)]
\item Assumption \ref{ass:stab-detect-lti} is satisfied.
\item $A_\f,B_\f$, and $C_\f$ are such that $A_\f$ is Hurwitz, $B_\f$ is full column rank, and $\left(\begin{smallmatrix}
A_\f - \lambda I_{n_\f} & B_\f\\
C_\f  & 0
\end{smallmatrix}\right)$ has full column rank for all $\lambda \in \C$ with $\Re(\lambda) \geq0$.
\end{enumerate}
\vspace{-0.4cm}Let $n_\cont=n_\p+n_\f$,  $A_\cont = A_\continuous - L_\continuous C_\continuous+B_\continuous K_\continuous$, $B_\cont=L_\continuous$ and $C_\cont=K_\continuous$ with $K_\continuous\in\R^{n_u\times n_\cont}$ such that 
$A_\continuous+B_\continuous K_\continuous$  Hurwitz and $L_\continuous\in\R^{n_\cont\times n_y}$  
such that $A_\continuous - L_\continuous C_\continuous$ is Hurwitz. The corresponding system $\Sigma_{\continuous}$  is ISS with respect to inputs $\widehat{y}_e$, $d$ and $w$.  
\end{propos}

\noindent\emph{Sketch of Proof.} The existence of such matrices $K_\continuous$ and $L_\continuous$ follow from Lemma \ref{lem:stab-detect-with-filter}, whose conditions hold. The desired result follows by application of the separation principle for LTI systems \cite[Theorem 16.10]{hespanha2009linear} and the fact that the global exponential stability of the origin for an LTI system implies it is ISS with respect to additive exogenous inputs \cite[Chapter 4.9]{khalil2002nonlinear}. \hfill $\Box$


\section{Illustrative example}\label{sect:example}

We apply the approach to stabilize a $4^{\text{th}}$-single-link manipulator to a given position using spiking communications between the sensors and the controller. The plant model is first presented   (Section \ref{subsect:example-system-objectives}). We then design the spike decoder  (\ref{eq:filter}) and  the controller (\ref{eq:controller}) (Section \ref{subsect:example-filter-controller}). We finally provide numerical simulations (Section \ref{subsect:example-simulations}).

\subsection{Model}\label{subsect:example-system-objectives}

We consider a single link manipulator with flexible joints and negligible damping   modeled as (\cite[Example~13.14]{khalil2002nonlinear})
\begin{equation}\label{eq:single-link-manipulator}
\begin{array}{rllllll}
\dot z_1 & = & z_2, & & 
\dot z_2 & = & -a\sin(z_1) - b(z_1-z_3)\\
\dot z_3 & = & z_4, & & 
\dot z_4 & = & c(z_1 - z_3) + u,
\end{array}
\end{equation}
where $z_1,z_3\in\R$ are angular positions with $z_1=z_3=0$ being the resting position, $z_2,z_4\in\R$ are the corresponding angular velocities, $u\in\R$ is a torque input, with parameters $a=0.1$, $b=1$, $c=1$. The goal is to stabilize the state $x^\star:=(
x_1^\star,0,x_3^\star,0)$ assuming $z=(z_1,\ldots,z_4)$ affected by additive measurement noise $w$ is available for control, hence $y = z+w$. For $x^\star$ to be an equilibrium point of (\ref{eq:single-link-manipulator}), it must hold that $x_3^\star = x_1^\star + \frac{a}{b}\sin(x_1^\star)$ and the input must be constant and verify  $u^\star:=\frac{ac}{b}\sin(x_1^\star)$. To write the plant model as in (\ref{eq:plant}), we take $x_\p=z-x^\star$ and $y=x_\p+x^\star+w$. As a result, system (\ref{eq:single-link-manipulator})  in the coordinates $x_\p$ becomes, with $x_\p=(x_{\p,1},\ldots,x_{\p,4})$, 

\begin{eqn}\label{eq:single-link-manipulator-x-coordinates}
\dot x_\p = A_\p x_\p + B_\p u + \psi(x_\p), & &
y = x_\p + x^\star+w,
\end{eqn}
with $A_\p=\left(\begin{smallmatrix} 0 & 1 & 0 & 0 \\ -b & 0 & b & 0 \\ 0 & 0 & 0 & 1\\ c & 0 & -c & 0\end{smallmatrix}\right)$,  $B_\p:=\left(0 \,\, 0 \,\, 0 \,\, 1\right)^\top$, $\psi(x_\p):=\big(0,-a\sin(x_{\p,1}+x_1^\star)- b (x_1^\star - x_3^\star),0,c(x_1^\star - x_3^\star)\big)$. 
System (\ref{eq:single-link-manipulator-x-coordinates}) is of the form of (\ref{eq:plant}) with $n_\p=4$, $n_u=1$, $n_d=0$,  $f_\p(x_\p,u,d)=A_\p x_\p+B_\p u +\psi(x_\p)$ and $h_\p(x_\p,w)=x_\p+x^\star+w$.

\subsection{Spike decoder}\label{subsect:example-filter-controller}

We design the spike decoder in (\ref{eq:filter}) such that $n_\f=4$, $A_\f=-\text{diag}(1,2,3,4)$, which is Hurwitz as required by Theorem \ref{thm:practical-iss-closed-loop}, $B_\f=A_\f$, and $C_\f = I_{n_\f}$. The corresponding $(x_\p,x_{\f,\n})$-system of $\Sigma_\continuous$ in (\ref{eq:sys-Sigma-continuous}) is given by 
\begin{eqn}\label{eq:example-model-and-filter}
\left(\begin{smallmatrix}\dot x_\p \\ \dot x_{\f,\n}  \end{smallmatrix}\right) & = & A_\continuous\left(\begin{smallmatrix} x_\p \\ x_{\f,\n}  \end{smallmatrix}\right) + B_\continuous u +  \Psi(x_\p,x_\f) + E_\continuous  w
\end{eqn}
with $A_\continuous:=\left(\begin{smallmatrix}A_\p & 0 \\ B_\f & A_\f \end{smallmatrix}\right)$, $B_\continuous:=\big(B_\p^\top \,\,\, 0^\top\big)$, $\Psi(x_\p,x_\f):=\big(\begin{smallmatrix}\psi(x_\p)^\top & (B_\f x^\star)^\top \end{smallmatrix}\big)^\top$ and $E_\continuous :=\big(\begin{smallmatrix} 0 & B_\f^\top \end{smallmatrix}\big)^\top$. We synthesize the controller so that the $x_\cont$-system of $\Sigma_\continuous$ in (\ref{eq:sys-Sigma-continuous}) reads, noting that $y_{\f,\n}=x_{\f,\n}$ and $y_{\f,e}=x_{\f,e}$ as $C_\f=I_{n_\f}$,
\begin{eqn}\label{eq:example-controller}
\dot x_\cont & = & A_\continuous x_\cont + B_\continuous u \!+\! \Psi(x_\cont) \\
& & + L_\continuous(x_{\f,\n}-x_{\f,e}-C_\continuous x_\cont)\\
u & = & \left(K  \,\,\,  0_{4\times 4}\right)x_\cont + u^\star,
\end{eqn}
where $n_\cont=8$,   $C_\continuous:=\left(0_{4\times 4} \,\,\, I_{4}\right)$, $K\in \R^{4\times 4}$ is such that $A_\p+B_\p K$ is Hurwitz and $L_\continuous\in\R^{8\times 4}$ is such that $A_\continuous-L_\continuous C_\continuous $ is Hurwitz. Such matrices $K$ and $L_\continuous$ exist as $(A_\p,B_\p)$ is stabilizable and $(A_\continuous,C_\continuous)$ is detectable, respectively. We select $K$ such that the spectrum of $A_\p +B_\p K$ is $\{-1,-2,-3,-4\}$ and we take  $L_\continuous$ such that the spectrum of $A_\continuous-L_\continuous C_\continuous $ is $\{-1,-2,\ldots,-8\}$. Controller (\ref{eq:example-controller}) is  an observer-based controller for system (\ref{eq:example-model-and-filter}) equipped with output $x_{\f,\n}$. 

System (\ref{eq:example-model-and-filter}), (\ref{eq:example-controller}) is ISS with respect to $x_{\f,e}$ and $w$. This property can be proved by using the  Lyapunov function $V(x_\p,x_\f)=(x_\p^\top\,x_\f^\top\,x_\cont^\top){\left(\begin{smallmatrix} P_1 + \vartheta P_2 & -\vartheta P_2\\ -\vartheta P_2 & \vartheta P_2  \end{smallmatrix}\right)}(x_\p,\,x_\f,x_\cont)$ where $P_1,P_2$ are real symmetric, positive definite matrices that  verify $(A_\continuous+B_\continuous\left(K \,\,\,  0_{4\times 4}\right))^\top P_1 + P_1(A_\continuous+B_\continuous\left(K \,\,\,  0_{4\times 4}\right))=-\text{diag}(\nu I_4,I_4)$ with $\nu>0$ sufficiently big, and 
$(A_\continuous-L_\continuous C_\continuous )^\top P_2 + P_2(A_\continuous-L_\continuous C_\continuous )=-5 I_8$, respectively, and $\vartheta>0$ is sufficiently big. As a result, Assumption \ref{ass:iss-Sigma-continuous} holds and Theorem \ref{thm:practical-iss-closed-loop} applies.

\subsection{Numerical simulations}\label{subsect:example-simulations}

We consider $x_1^\star=3\tfrac{\pi}{2}$. Thus $x_3^\star = x_1^\star + \frac{a}{b}\sin(x_1^\star) = 3\tfrac{\pi}{2} + 0.1\sin(3\tfrac{\pi}{2})$. 
The measurement noise is taken to be $(0,0,0.2\sin(2\pi f_1 t)+0.5\cos(2\pi f_2 t),0.1\sin(f_1 t+\pi/4))$ with $f_1=20$ Hz and $f_2=500$ Hz. 
Regarding the neurons, as $n_\p=4$ and the full state is available for control,  $n_y=n_\p$ and  we implement $2n_\p=8$ membrane potentials: two per component of $x_\p$ consistently with Section~\ref{subsect:spiking-encoder}. We have selected all the constants $\Delta_{i,\ell}  = \alpha_{i,\ell}$, $i\in\{1,\ldots,4\}$ and $\ell\in\{1,2\}$, equal to $\Delta$ and we have run simulations for different values of $\Delta$ in $\{0.01, 0.1, 1\}$. The initial conditions are $x_\p(0)=-x^\star$, i.e., the manipulator is initialized at the resting down position with zero velocities, $x_\f(0)=0_{4\times 1}$, $x_\cont(0)=0_{8\times 1}$ and $\xi(0)=0_{8\times 1}$.

Fig. \ref{fig:norm-xp} represents the norm of $x_\p$ over the interval $[0,10]$ with: \emph{(i)} continuous  communications  corresponding to the nominal case where $y$ is communicated at all time instants to the spike decoder, \emph{(ii)} spiking communications for the different values of $\Delta$. We observe that the smaller $\Delta$, the smaller the ultimate bound on the norm of $x_\p$, which is in line with Theorem \ref{thm:practical-iss-closed-loop} where $|\alpha|=\sqrt{2n_\p} \Delta$. Fig. \ref{fig:spikes} illustrates the spiking communication for each state component for the case where $\Delta=0.1$. Finally, to evaluate the trade-off between performance  and amount of communication, we considered  the  average number of spikes per time unit, i.e., the total number of spikes divided by the simulation time, as well as the  ultimate bound on $x_\p$ evaluated by taking $\|x_\p\|_{[9,10]}$ both averaged over $100$ different initial conditions, see Table \ref{tab:example-spiking-rate-performance}. Specifically, we considered $100$ initial conditions of the form $(x_{\p,1,0},0,x_{\p,3,0},0)$ with $x_{\p,i,0}$ taking 10 possible values equally spaced in the interval $[0,2\pi]$ for $i\in\{1,3\}$.  
Table \ref{tab:example-spiking-rate-performance} confirms the intuition that more spikes lead to better  performance, here in terms of $\|x_\p\|_{[9,10]}$. Interestingly, when the output $y$ is continuously communicated to the spike decoder, the ultimate bound of $x_\p$ is larger than with spiking communications with $\Delta=0.01$. This may be explained by the fact that with continuous communication the noisy measured output is transmitted continuously to the filter and controller and the amplitude of the noise directly influences the closed-loop performance. In contrast, with spiking communications, the noisy output is input to the neurons. This provides additional filtering; the noise may produce a time difference in the occurrence of the fixed-amplitude spikes, but it does not directly affect the signal amplitude. 

\begin{figure}
\centering
\includegraphics[scale=0.4]{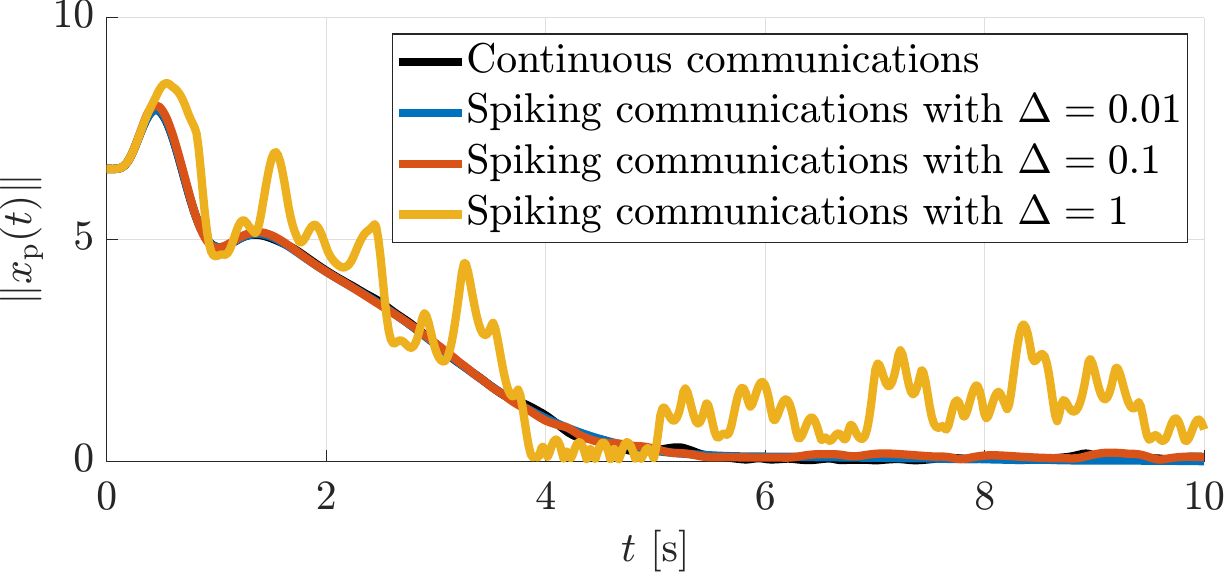}
\caption{Norm of $x_\p$ for the considered set-up with either continuous communications or spiking communications and different values of $\Delta$.}\label{fig:norm-xp}
\end{figure}

\begin{figure}
\centering
\includegraphics[scale=0.4]{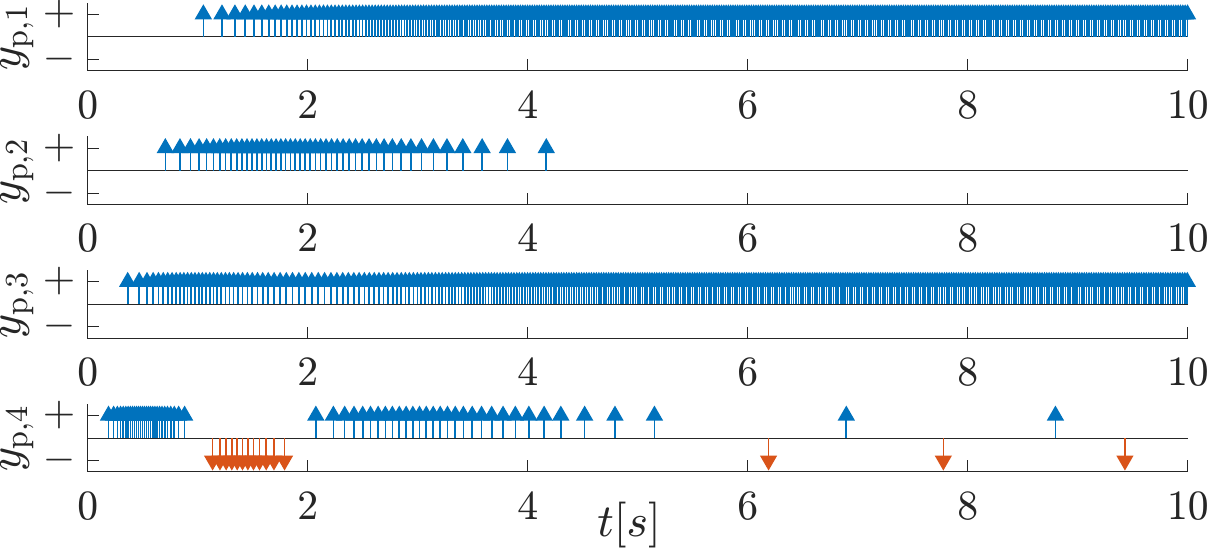}
\caption{Spiking communications for each output component for $\Delta=0.1$.}\label{fig:spikes}
\end{figure}



\begin{table}[ht]
\centering
\label{tab:spike_rates}
\begin{tabular}{lll}
$\Delta$ & 
Average spiking rate & Average $\|x_\p\|_{[9,10]}$ \\
\midrule
0.01 & 1130.1 & 0.021 \\
0.1  & 113.1  & 0.150 \\
1   &  11.66  & 1.472 \\
\text{Continuous} & \text{n/a} & 0.131 \\
\midrule
\end{tabular}
\caption{Average spiking rate vs. average $\|x_\p\|_{[9,10]}$ over $100$ initial conditions and different values of $\Delta$ as well as for  continuous-time communication.}
\label{tab:example-spiking-rate-performance}
\end{table}
\section{Conclusions}\label{Conclusion}

We proposed a framework for designing stabilizing output-feedback controllers using neuronal spiking communications. The plant output is encoded into a spiking signal and decoded via a synaptic filter for control. We established general design conditions on the controller, encoder, and decoder to guarantee closed-loop practical ISS. 
We believe this work opens the door to many developments, including addressing 
distributed scenarios, or application to formal analysis of control loops including $\Sigma-\Delta$ modulators, see Remark \ref{rem:SigmaDelta}, or to address  other control objectives, such as the stabilization of limit cycles, see, e.g., \cite{schmetterling2024neuromorphic, medvedeva2025formalizing, petri2025rhythmic}.

\bibliographystyle{plain}        
\bibliography{bibliography}           

\appendix

\section{Intermediate results} \label{Appendix:GeneralizationTACemulation} 

This appendix presents several results needed to prove Proposition \ref{thm:practical-iss-Sigma-s} in Appendix \ref{Appendix:ProofThmBoundOnSpikyFilterOutput}. These results   revisit some properties established in \cite{petri2025emulation}, the difference is that membrane potentials are allowed to be initialized at their threshold value, i.e., jumps at the initial time $t=0$ are allowed contrary to \cite{petri2025emulation}.  

\begin{propos}\label{Prop:InterSpikingTime} Let $(i,\ell)\in\{1,\ldots,n_y\}\times\{1,2\}$,  $y_i\in\mathcal{L}_{\R}$ and consider neuron $(i,\ell)$ in \eqref{eq:thresholdFlow}-\eqref{eq:SpikingMeasurement} with input $y_i$. Any solution with  initial condition $\phi_{\xi_{i,\ell}}(0)\in[0,\Delta_{i,\ell}]$ is such that the corresponding sequence of spiking times is in $\mathcal{T}_{\zf}^\infty$. 
\end{propos}

\noindent\emph{Sketch of Proof.} The proof follows the same steps as those in the proof of \cite[Proposition 1]{petri2025emulation}. The only difference is that we may have  at most $2n_y$ spikes before \cite[(43)]{petri2025emulation} applies, but the proof reasoning and the conclusion remain the same. \hfill $\Box$

The next proposition  provides  a universal approximation property of the $2n_y$-neuron network in Fig.~\ref{Fig:blockDiagram_ASconvMultipleOutput}. 

\begin{propos}\label{Prop:UniversalApproximationMultiple}
    Given any $y\in\mathcal{L}_{\R^{n_y}}$, consider the $2n_y$-neuron integrate-and-fire spiking neuronal network in  \eqref{eq:thresholdFlow}-\eqref{eq:SpikingMeasurement} with $\alpha_{i,\ell}=\Delta_{i,\ell}$ for any $(i,\ell)\in\{1,\ldots,n_y\}\times\{1,2\}$. 
    For any solution $\phi_\xi$ initialized in $[0,\alpha_{1,1}]\times\ldots\times [0,\alpha_{n_y,2}]$, with input $y$, then $e:=y-y_s$ satisfies for any $t \in \R_{\geq0}$, 
		$|\int_{0}^{t}  e(s) ds| \leq 2\sqrt{n_y}|\alpha|$.
\end{propos}

\begin{proof} Let $y\in\mathcal{L}_{\R^{n_y}}$ and $\phi_\xi$ be a solution to  \eqref{eq:thresholdFlow}-\eqref{eq:SpikingMeasurement} as specified in Proposition \ref{Prop:UniversalApproximationMultiple}. Consider $i\in\{1,\ldots,n_y\}$. When $\phi_{\xi_{i,\ell}}(0)\in[0,\Delta_{i,\ell})$ for all $\ell\in\{1,2\}$, we have from \cite[Thm. 1]{petri2025emulation} (noting that $K_\ell=1$ here  with the notation of \cite{petri2025emulation}), for all $t\geq 0$,
\begin{equation}
\textstyle\left|\int_{0}^{t}  e_i(s) ds\right| \leq \alpha_{i,1} + \alpha_{i,2}, 
\label{eq:boundEmulationErrorTheoremUsedINProofProofMultiple}
\end{equation}
with $e_i(t)$ is the $i^{\text{th}}$ component of $e=y-y_s$. 
When there exists $\ell\in\{1,2\}$ such that $\phi_{\xi_{i,\ell}}(0)=\Delta_{i,\ell}$ for some $\ell\in\{1,2\}$, the right-hand side of \cite[(54)]{petri2025emulation} becomes $K_\ell(-\Delta_\ell+\xi_{\ell}(t))=K_\ell(-\xi_\ell(0)+\xi_{\ell}(t))$, which corresponds to   \cite[(54)]{petri2025emulation}. As a result,   the rest of the proof of \cite[Thm. 1]{petri2025emulation}  holds. Consequently, (\ref{eq:boundEmulationErrorTheoremUsedINProofProofMultiple}) also holds when $\phi_{\xi_{i,\ell}}(0)=\Delta_{i,\ell}$ for some $\ell\in\{1,2\}$.  
As $|\int_{0}^{t}e(s)ds|\leq \sum_{i=1}^{n_y}|\int_{0}^{t}  e_i(s) ds|$, we derive that 
$|\int_{0}^{t}  e(s) ds| \leq \sum_{i=1}^{n_y}(\alpha_{i,1} + \alpha_{i,2})\leq 2\sqrt{n_y}|\alpha|$.  
\end{proof}

The next theorem establishes that any LTI system with 
is practically spiking ISS; see \cite{petri2025emulation}. 

\begin{theorem}\label{Thm:SISS_generalized}
Consider the LTI system $\dot z= Fz + Gv$, with $z \in \R^{n_z}$, $v \in \mathcal{S}_{n_v}$, $n_z, n_v \in \Np$,  $F \in \R^{n_z \times n_z}$ Hurwitz and $G\in\R^{n_z\times n_v}$. There exist $\beta_z \in \exp$-$\KL$  and $c_z \in \R_{>0}$ such that for any  $v \in \mathcal{S}_{n_v}$, any solution $\phi_z$ with input $v$ is defined for all $t \geq0$ and satisfies
\begin{equation}\label{eq:sISS_generalized}
|\phi_z(t)| \leq \beta_z(|\phi_z(0)|,t) + \tilde{\gamma}_z \norm{v}_\star + c_z \norm{A_0}, 
\end{equation}
for all $t \geq0$, with  $\tilde{\gamma}_z:= \norm{G} + \int_{0}^{\infty} \norm{F e^{Fs}G} ds \in \R_{\geq 0}$ and $A_0 \in \R^{n_v}_{\geq0}$  the spike amplitude of $v$ corresponding to time $0$ (see Section \ref{sect:preliminaries}). 
\end{theorem}

\begin{proof} The proof follows similar steps as the proof of \cite[Thm.~2]{petri2025emulation}. We therefore only highlight the differences.  
 Let $v \in \mathcal{S}_{n_v}$ and $\phi_z$ be a solution to the LTI system in Theorem \ref{Thm:SISS_generalized} and let $t\geq 0$. When $t_0<t_1$, i.e., no spike occurs at the initial time,  \cite[Thm.~2]{petri2025emulation} applies and (\ref{eq:sISS_generalized}) holds. Consider now the case where $t_1=t_0=0$, i.e., a spike occurs at time $0$. By following similar steps as in \cite[proof of Thm.~2]{petri2025emulation}, 
\begin{equation}
	\begin{aligned}
		&\textstyle\int_{0}^{t} h(t-s) v(s) ds = \int_{0}^{t_{1}^+} h(t-s) v(s) ds \\ 
        &\textstyle+ \sum\nolimits_{i = 2}^{J} \Big(\int_{t_{i-1}^+}^{t_{i}^-} h(t-s) v(s) ds 
		+ \int_{t_{i}^-}^{t_{i}^+} h(t-s) v(s) ds \Big) \\
        &\textstyle+ \int_{t_{J}^+}^{t} h(t-s) v(s) ds, 
	\end{aligned}
	\label{eq:StrongIISSLTI_integralDifferentComponents}
\end{equation}
instead of \cite[(63)]{petri2025emulation}, where $h(\theta):= e^{F\theta}G \in \R^{n_z\times n_v}$, $\theta \in \R$, and $\{t_i\}_{i \in \N_{\leq J}} \in \mathcal{T}_{\zf}^\infty$ the spiking times associated to $v$. Since $t_1= 0$,  $\int_{0}^{t_{1}^+} h(t-s) v(s) ds = h(t)A_0$. 
Thus \cite[(69)]{petri2025emulation} becomes 
\begin{equation}
\begin{aligned}
		&\textstyle\int_{0}^{t} h(t-s)v(s) ds 
		= h(t)A_0 \\ 
        &\textstyle+
        \sum\nolimits_{i= 2}^{J} \Big(h(t-t_{i})\int_{0}^{t_{i}^-} v(\theta)d\theta 
		- h(t-t_{i-1})\times\\ 
		&\textstyle\int_{0}^{t_{i-1}^+} v(\theta)d\theta   
		-\int_{t_{i-1}^+}^{t_{i}^-} \frac{d}{ds}\left(h(t-s)\right) \int_{0}^s v(\theta)d\theta ds 
		\\ &\textstyle+ h(t-t_{i}) A_{i} \Big)
		+ h(0)\int_{0}^{t_1^+} v(\theta)d\theta + h(0)\int_{t_1^+}^{t} v(\theta)d\theta \\
        &\textstyle- h(t-t_{j(t)})
		\int_{0}^{t_{J}^+} v(\theta)d\theta 
		- \int_{t_{J}^+}^{t} \frac{d}{ds}\left(h(t-s)\right) \\
        &\textstyle\times\int_0^s v(\theta)d\theta ds,  
	\end{aligned}
	\label{eq:StrongIISSLTI_integralDifferentComponents2temp}
\end{equation}
where we recall that $\{A_i\}_{i \in \N_{\leq J}}$ denotes the spike amplitudes associated to $v$, as defined in Section \ref{sect:preliminaries}, and the term $h(0)\int_{0}^{t} v(\theta)d\theta$ in \cite[(69)]{petri2025emulation} is replaced by $h(0)\int_{0}^{t_1^+} v(\theta)d\theta + h(0)\int_{t_1^+}^{t} v(\theta)d\theta = h(0)A_0 +  + h(0)\int_{t_1^+}^{t} v(\theta)d\theta$. Thus, following similar steps as in \cite[(70)-(72)]{petri2025emulation}, \cite[(72)]{petri2025emulation} becomes
\begin{equation}
	\begin{aligned}
		&\textstyle\int_{0}^{t} h(t-s) v(s) ds  
		= (h(t) + h(0))A_0 + h(0)\int_{t_1^+}^{t} v(\theta)d\theta \\
        &\textstyle- \int_{0}^{t} \frac{d}{ds}\left(h(t-s)\right) \int_0^s v(\theta)d\theta ds \\
		 &\textstyle
         +\sum_{i = 2}^{J} h(t-t_{i})\Big(
		- \int_{t_{i}^-}^{t_{i}^+} v_{1}(\theta)d\theta\Big), 
	\end{aligned}
\end{equation} 
which implies 
\begin{equation}
\begin{aligned}
	\textstyle\int_{0}^{t} h(t-s) v(s) ds 
	=  (h(t) + h(0))A_0 +  h(0)\int_{t_1^+}^{t} v(\theta)d\theta \\
	\textstyle- \int_{0}^{t} \frac{d}{ds}\left(h(t-s)\right) \int_0^s v(\theta)d\theta ds, 
\end{aligned}
\label{eq:StrongIISSLTI_integralDifferentComponents4}
\end{equation}
instead of \cite[(73)]{petri2025emulation}. Recalling that $h(\theta) = e^{F\theta}G$ for any $\theta \in \R$, from \eqref{eq:StrongIISSLTI_integralDifferentComponents4} instead of \cite[(74)]{petri2025emulation} we have
\begin{equation}
\begin{aligned}
	\textstyle\phi_{z}(t)&\textstyle=e^{F t}\phi_{z}(0) + \int_{0}^{t} e^{F (t-s)} G v(s) ds\\
    &\textstyle  =  e^{F t}\phi_{z}(0) + (e^{Ft}G+G)A_0 + G\int_{t_1^+}^{t} v(\theta)d\theta \\
    & \quad \textstyle    - \int_{0}^{t} (-Fe^{F(t-s)}G)
	 \int_0^s v(\theta)d\theta ds. 
\end{aligned}
\label{eq:StrongIISSLTI_stateEquation2}
\end{equation}
Since $F \in \R^{n_z\times n_z}$ is Hurwitz, there exists $c_0\in \R_{>0}$ such that $\norm{e^{F\tau}G+G} \leq c_0$ for any $\tau\geq 0$. Thus, there exist $c_z \in  \R_{>0}$ such that, following similar steps as in \cite[Proof of Thm. 2]{petri2025emulation}, 
	$|\phi_{z}(t)| 
	 \leq \beta_z(|\phi_{z}(0)|, t) + \gamma_z(\norm{v}_\star) + c_z \norm{A_0}$,  
with $\beta_z \in \exp$-$\KL$ and $\tilde\gamma_z \in\Rlp$  as defined in Theorem \ref{Thm:SISS_generalized}, which are both independent of initial conditions, input $v$ and time $t$. 
\end{proof}

\section{Proof of Proposition \ref{thm:practical-iss-Sigma-s}}\label{Appendix:ProofThmBoundOnSpikyFilterOutput}

Let $y\in\mathcal{L}_{\R^{n_y}}$ and $\phi_s = (\phi_{\f,e}, \phi_\xi)$ be a solution to \eqref{eq:sys-Sigma-s} with input $y$ and initial conditions as in Proposition \ref{thm:practical-iss-Sigma-s} and let $t\geq 0$. 
By Proposition \ref{Prop:InterSpikingTime}, the sequence of spiking times 
$\{t_{j}^{(i,\ell)}\}_{j\in\N_{\leq J_{i,\ell}}} \in \mathcal{T}_{\zf}^\infty$, with $J_{i,\ell}\in\No\cup\{\infty\}$. As the number of neurons is finite,  the  sequence of spiking times of $\phi_s$ also belongs to $\mathcal{T}_{\zf}^\infty$. Consequently, since the continuous-time dynamics of $\phi_s$ is defined by a globally Lipschitz map,  $\phi_s$ is complete thereby proving Proposition \ref{thm:practical-iss-Sigma-s}\emph{(i)}  as $\phi_s$ has been arbitrarily selected. By Proposition \ref{Prop:UniversalApproximationMultiple}, 
		$\textstyle\norm{ e}_\star = \sup_{t \in \R_{\geq0}} |\int_{0}^{t} e(s) ds| \leq 2\sqrt{n_y}|\alpha|$.  
		 \label{eq:boundEmulationErrorTheoremMultipleProofTh2}
Thus $e \in \mathcal{S}^{n_y}$.
Moreover, since $A_\f$ in \eqref{eq:filter-x-f-e} is Hurwitz and the norm of the spike amplitudes associated to $e$ are bounded by $|\alpha|$, Theorem \ref{Thm:SISS_generalized} ensures that  
\begin{equation}\label{eq:sISSfilter}
    |\phi_{\f,e}(t)| \leq \beta_z(|\phi_{\f,e}(0)|,t) + \tilde\gamma_z \norm{e}_\star + c_z|\alpha|
\end{equation}
with $\beta_z\in\exp$-$\KL$, $\tilde\gamma_z := \norm{B_\f} + \int_{0}^{\infty}\norm{A_\f e^{A_\f s} B_\f} ds $ and $c_z>0$. 
We derive
    $|\phi_{\f,e}(t)| 
    \leq \beta_z(|\phi_{\f,e}(0)|,t) + (2\tilde\gamma_z\sqrt{n_y}+c_z)|\alpha|$. 
Therefore
$
  |\phi_{\f,e}(t)| \leq \max\{2\beta_z(|(\phi_{\f,e}(0)|,t), 2(c_z+2\tilde\gamma_z\sqrt{n_y})|\alpha|\}=\max\{\beta_s(\phi_{\f,e}(0)|,t),\tilde\gamma_s|\alpha|\}$ with $\beta_s:=2\beta_z\in\exp$-$\KL$ and $\tilde \gamma_s:=2(c_z+2\tilde\gamma_z\sqrt{n_y})$.
The result is obtained by noting  $\phi_{\xi}(t)\in [0, \alpha_{1,1}]\times \dots \times [0, \alpha_{n_y,2}]$, so that  $|\phi_s(t)|_{\mathcal{A}^{s}_{\alpha}}= |\phi_{\f,e}(t)|$ and $|\phi_s(0)|_{\mathcal{A}^{s}_{\alpha}}= |\phi_{\f,e}(0)|$. 



\section{Proof of Proposition \ref{prop:stability-augmented-closed-loop}}\label{appendix:proof-prop-stability-augmented-closed-loop}

Let $d\in\mathcal{L}_{\R^{n_d}}$, $w\in\mathcal{L}_{\R^{n_w}}$ and  $\phi^\aug$ be a  solution to  (\ref{eq:closed-loop-augmented}) with inputs $d$ and $w$ as specified in Proposition \ref{prop:stability-augmented-closed-loop}. Let $[0,t^\star)$ with $t^\star\in\Rlp\cup\{\infty\}$ be the interval over which the solution is defined, and take $t\in[0,t^\star)$. Denote $\phi_{\continuous}^\aug:=(\phi_{\p}^\aug,\phi_{\f,\n}^\aug,\phi_{\cont}^\aug)$. By Assumption \ref{ass:iss-Sigma-continuous}, there exist $\beta_\continuous\in\KL$ and $\gamma_\continuous\in\Kinf$ independent of $d$, $w$, $t$ and $\phi^\aug$ such that
\begin{eqn}\label{eq:proof-prop-stability-augmented-iss-Sigma-continuous}
|\phi_{\continuous}^\aug(t)| & \leq &  \max\big\{\beta_\continuous(|\phi_{\continuous}^\aug(0)|,t), \gamma_\continuous(\|\widehat{y}_e\|_{[0,t]}),\\
& &\quad\quad\,\,\,\gamma_\continuous(\|d\|_{[0,t]}),\gamma_\continuous(\|w\|_{[0,t]})\big\}
\end{eqn}
where $\widehat{y}_e=C_\f \phi_{\f,e}$. We derive from Proposition~\ref{thm:practical-iss-Sigma-s}, the facts that $\phi_{\f,e}(0)=0$, and $\phi_{\xi_{i,\ell}}(0)\in[0,\Delta_{i,\ell}]$ with $\alpha_{i,\ell}=\Delta_{i,\ell}$ and the definition of $\alpha$,
$|\widehat{y}_{e}(t)| \leq   \|C_\f\|\max\big\{\beta_s(|\alpha|,t),\tilde{\gamma}_s|\alpha|\big\}$. 
Hence, as $\beta_s\in\KL$, 
$|\widehat{y}_{e}(t)| \leq  \|C_\f\|\max\big\{\beta_s(|\alpha|,0),\tilde{\gamma}_s|\alpha|\big\}$. 
Thus 
\begin{eqn}
\|\widehat{y}_{e}\|_{[0,t]} & \leq &  \|C_\f\|\max\big\{\beta_s(|\alpha|,0),\tilde{\gamma}_s(|\alpha|)\big\}=:\overline\gamma_s(|\alpha|),
\end{eqn}
We use this last inequality in (\ref{eq:proof-prop-stability-augmented-iss-Sigma-continuous}) to derive 
\begin{eqn}\label{eq:proof-prop-stability-augmented-iss}
|\phi_{\continuous}^\aug(t)| & \leq &  \max\big\{\beta_\continuous(|\phi_{\continuous}^\aug(0)|,t),\gamma_\continuous\circ\overline\gamma_s(|\alpha|),\\
& & \quad\quad \quad\gamma_\continuous(\|d\|_{[0,t]}),\gamma_\continuous(\|w\|_{[0,t]})\big\}.
\end{eqn}
We have $|\phi^\aug(t)|_{\mathcal{A}_\alpha^\aug}\leq  |\phi_{\continuous}^\aug(t)| + |(\phi_{\f,e}^\aug(t),\phi_\xi^{\aug}(t))|_{\mathcal{A}_\alpha^s}$ with $\mathcal{A}_\alpha^s$ in (\ref{eq:As_set}) and $\mathcal{A}_\alpha^\aug$ in \eqref{eq:X-aug-A-aug}. Thus, by (\ref{eq:proof-prop-stability-augmented-iss}) and Proposition \ref{thm:practical-iss-Sigma-s}\emph{(ii)}, for all $t\in[0,t^\star)$,
\begin{eqn}\label{eq:proof-prop-stability-augmented-iss-overall}
|\phi^\aug(t)|_{\mathcal{A}_{\alpha}^{\aug}} & \leq & \max\big\{\beta_\aug(|\phi^\aug(0)|_{\mathcal{A}_{\alpha}^{\aug}},t),\\
& & 
\gamma_\aug(\|d\|_{[0,t]}),\gamma_\aug(\|w\|_{[0,t]}),\gamma_\aug(|\alpha|)\big\},
\end{eqn}
where $\beta_\aug=2\max\{\beta_\continuous,\beta_s\}\in\KL$ and $\gamma_\aug=2\max\{\gamma_\continuous\circ\overline{\gamma}_s,\gamma_\continuous,\tilde\gamma_s \id\}\in\Kinf$. 
We need to show that $t^\star=\infty$ to prove the completeness of solutions.
The two possible obstacles for $t^\star=\infty$ are Zeno behavior or finite escape times. Equation (\ref{eq:proof-prop-stability-augmented-iss-overall}) implies that $\phi^\aug$ cannot escape to infinity over $[0,t^\star)$ as the set $\mathcal{A}_{\alpha}^{\aug}$ in (\ref{eq:X-aug-A-aug}) is compact, $(d,w)\in\mathcal{L}_{\R^{n_d}}\times\mathcal{L}_{\R^{n_w}}$ and $\gamma_{\text{aug}} \in \Kinf$. This implies that  $\|\phi_{\continuous}^\aug\|_{[0,t^\star]}$ is bounded, thus so is $\|y\|_{[0,t^\star]}$ as $y=h_\p(\phi_\p,w)$ and $h_\p$ is continuous. To also exclude Zeno, consider the signal $\overline y:\Rlo\to\R^{n_y}$ defined as $\overline y=y$ on $[0,t^\star)$ and $\overline y=0$ on $[t^\star,\infty)$. It holds that $\overline y\in\mathcal{L}_{\R^{n_y}}$. Let $\overline\phi_\xi:\Rlo\to\R^{2n_y}$ be defined as $\overline\phi_\xi=\phi_\xi$ on $[0,t^\star)$ and as the solution to the $\xi$-system in (\ref{eq:closed-loop}) initialized at $\phi_\xi(t^\star)$, with input $\overline y=0$ on $[t^\star,\infty)$. By  Proposition \ref{Prop:InterSpikingTime}, 
 the sequence of spiking times associated to $\overline\phi_\xi$ belongs to $\mathcal{T}_{\zf}^\infty$. 
 As the spiking times of $\overline\phi_\xi$ match those of  $\phi$ up until $t=t^\star$, we reach a contradiction. Hence, no Zeno for $\phi$ and $t^\star=\infty$.

\end{document}